\documentclass{article}
\usepackage[british]{babel}
\usepackage{amsfonts}
\usepackage{amsmath}
\usepackage{color}
\usepackage{hyperref}
\usepackage{graphicx}
\usepackage{float}
\newtheorem{theorem}{Theorem}

\newtheorem{corollary}[theorem]{Corollary}

\newtheorem{example}[theorem]{Example}

\newtheorem{proposition}[theorem]{Proposition}
\newtheorem{remark}[theorem]{Remark}

\newtheorem{model}[theorem]{Model}

\newenvironment{proof}[1][Proof]{\noindent\textbf{#1.} }{\ \rule{0.5em}{0.5em}}

\newcommand{\PP}{\mathbb{P}}

\begin{document}

\title{On smile properties of volatility derivatives and exotic products: understanding the VIX skew}
\author{Elisa Al\`{o}s \thanks
{Supported by grant MTM2016-76420-P.}
 \\
Dpt. d'Economia i Empresa\\
Universitat Pompeu Fabra\\
and Barcelona GSE\\
c/Ramon Trias Fargas, 25-27\\
08005 Barcelona, Spain 
\and David Garc\'ia-Lorite\\
CaixaBank\\
and Universitat de Barcelona\\
Av. Diagonal, 621-629, 08028 - Barcelona
\and Aitor Muguruza \thanks
{Supported by Centre for Doctoral Training in Financial Computing \& Analytics.}\\
NATIXIS\\
 and Department of Mathematics\\
Imperial College London\\
London SW7 2AZ\\}
\date{}

\maketitle

\begin{abstract}
We develop a method to study the implied volatility for exotic options and volatility derivatives with European payoffs such as VIX options. Our approach, based on Malliavin calculus techniques, allows us to describe the properties of the at-the-money  implied volatility (ATMI) in terms of the Malliavin derivatives of the underlying process. More precisely, we study the short-time behaviour of the ATMI level and skew. As an application, we describe the short-term behavior of the ATMI of VIX and realized variance options in terms of the Hurst parameter of the model, and most importantly we describe the class of volatility processes that generate a positive skew for the VIX implied volatility. In addition, we find that our ATMI asymptotic formulae perform very well even for large maturities. Several numerical examples are provided to support our theoretical results.

\vspace{0.2cm}


\vspace{0.2cm}


\end{abstract}

\section{Introduction}
La raison d' \^{e}tre of any financial model is to reproduce some behaviour or dynamic that is observed in the market. For instance, in the case of vanilla options this conundrum is translated into fitting the market implied volatility surface. To this end, several extensions of the Black-Scholes model have been proposed. In particular, one that effectively achieves this task is the Dupire local volatility model \cite{Dupire}. Nevertheless, when incorporating exotic products into the vanilla universe, one finds that the local volatility model does not properly reproduce the dynamics of the market (see \cite{SABR} for details on barrier options). A different and popular approach is to allow the volatility to be a stochastic process itself (see for example Hull and White \cite{HullWhite},  Scott \cite{Scott}, Heston \cite{Heston}, Stein and Stein \cite{Stein} and Ball and Roma \cite{BallRoma}). It is well known that classical stochastic volatility (SV) models, where the volatility follows a diffusion process, can explain some important properties of the implied volatility, as its variation with respect to the strike price, described graphically as the smile or skew (see Renault and Touzi \cite{RenaultTouzi}) and the leverage effect. Most importantly, aside from static properties they provide realistic dynamics of the spot in order to price exotic products. In spite of all this facts, this first generation of SV models does not capture some other important features of market data, as the term structure (dependence on the time to maturity) of the implied volatility. In an attempt to fix this issue, Bergomi \cite{Bergomi} introduced a second generation of SV models coupled with the concept of forward variance, which allows time dependent structures of forward volatilities. Then one may successfully solve the calibration problem by considering a local stochastic volatility (LSV) model as suggested by Lipton \cite{Lipton} and obtain a model that both fits the implied volatility surface and has realistic dynamics on the spot.
\\\\
In spite of the popularity of LSV models, the addition of volatility derivatives in the picture dramatically complicates the calibration and pricing. Concretely, in a LSV model the value of a log-contract (or idealised variance swap) is given by
\begin{equation}\label{log-contract}
\frac{1}{T}E\int_0^T v_u (\sigma(S_u,u))^2 du, \quad T>0
\end{equation}
where $v$ represents the pure stochastic volatility process, $S$ represents the spot and $\sigma(\cdot, \cdot)$ is the local volatility component. In this setting, the sole pricing of VIX options or even futures (involving conditional expectations and nonlinearities) is extremely involved. This is the reason why Bergomi \cite{Bergomi} type of SV models are gaining popularity i.e.  their friendly structure of the log-contract in \eqref{log-contract} with $\sigma(\cdot, \cdot)=1$.
\\\\
With this picture in mind, the main challenge on equity markets is to jointly fit spot vanilla and VIX smiles, whilst providing realistic dynamics on both, spot and volatility processes. In order to construct models that allow us to describe this complexity, it is important to develop tools that allow us to identify the  class of models that are able to generate the desired behaviour. One first step in this direction was developed in Al\`os, Le\'on and Vives \cite{AlosLeonVives}, where the authors described the short-time behavior of the at-the-money implied volatility (ATMI) skew  in terms of the Malliavin derivative operator of the SV process. This result establishes that the  blow-up in this slope ( observed in real market data), can be described by a volatility process $\sigma$ such that $D_s\sigma_r\to \infty$ as $s\to r$, where $D$ denotes the Malliavin derivative operator (see for example Nualart (2005)). This property is satisfied by stochastic volatility models based on the fractional Brownian motion (fBm) with Hurst parameter $H<\frac12$.
\\\\
This observation has lead to the recent development of rough volatility models (see for example Fukasawa \cite{Fukasawa} or Bayer, Friz and Gatheral \cite{BFG15}). Remarkably, rough volatility models not only provide a realistic implied volatility surface, but also agree with the historical dynamics of volatility as shown by Gatheral, Jaisson and Rosenbaum \cite{GJR18} and Bennedsen, Lunde and Pakkanen \cite{BLP16}. We must emphasize here that neither of the first nor second generation of SV models were found to agree with the historical dynamics of volatility. Not surprisingly though, there is a price to pay for such innovative model; Markovianity. In order to overcome the lack of Markovianity and the lost of classical tools (such as PDE's and It\^o's lemma) a vast literature has emerged from this active area of research (see \cite{AlosShiraya}, \cite{BFGHS}, \cite{GJR14}, \cite{Neuman}, \cite{JMM17}, \cite{MP17}, \cite{HJL} and \cite{JPS17} among others).
\\\\
Our main objective in this paper is to study the ATMI short-time level and skew of realized variance (RV) options and VIX options. As a first step, we will see that exotic option prices on a underlying $A$ coincide with vanilla option prices where the underlying is a SV model where the volatility is determined by the Malliavin derivative of the underlying process $A$. This will allow us to apply previous results on the implied volatility level and skew (see Al\`os and Shiraya \cite{AlosShiraya} and Al\`os, Le\'on and Vives \cite{AlosLeonVives}). Our results provide a method based on the techniques of Malliavin calculus to estimate the ATMI rate of the short-dated level and skew. In particular, we will see that, if $D_s\sigma_r=O(s-r)^{H-\frac12}$, for some $H\in \left(\frac12,0\right)$, the ATMI level and skew are of the order $O((T)^{H-\frac12})$ for RV options and of the order $O(1)$ for VIX options. Moreover, we develop an easy-to-apply criteria to determine the class of stochastic volatility models such that the corresponding VIX skew is positive, as observed in real market data. This simple tool allows us to check that the VIX skew is negative for the Heston model and zero for the SABR, while it becomes positive with a "mixing lognormals" solution. This coincides with previous results (see Baldeaux and Badran \cite{BB} and Bergomi \cite{Bergomi3}).
\\\\
VIX options and futures are becoming increasingly popular both in the industry and academic research (see \cite{CarrMadan} and \cite{HJT} for instance). We find that our results along with Bergomi \cite{Bergomi3} and De Marco \cite{DeMarco} (in the rough volatility context)  combine the necessary conditions to generate models with positive VIX skew in forward variance form. Concretely, Proposition \ref{prop:fbmATMISkewVIX} gives such conditions for a large family of models covered by both authors. In addition, we find that our asymptotic formulas yield accurate approximation for maturities up to 6 months that could help understand the joint dynamics of VIX and S\&P 500, also analysed by Guyon \cite{Guyon} using a different approach.
\\\\
This paper is organized as follows. In Section 2 we introduce some basic concepts on Malliavin calculus. Section 3 is devoted to see how exotic option prices on a underlying $A$ coincide with vanilla option prices where the underlying is a stochastic volatiliy model where the volatility is determined by the Malliavin derivative of the process $A$. As a consequence, in Section 4, we obtain our results for the ATMI level and skew. In section 5, we present a family of models that we will use in subsequent sections. Finally, the cases of VIX and RV options are studied in Sections 6 and 7, respectively.

\section{Preliminaries on Malliavin calculus}
We assume that the reader is familiar with the elementary notions of
Malliavin calculus, as given for instance in Nualart \cite{Nualart}. Consider a Brownian motion $W$ defined on a probability space $(\Omega, \mathcal{F},\mathbb{P})$. The set $\mathbb{D}_{W}^{1,2}$ will denote the domain of the derivative operator $%
D$ with respect to the Brownian Motion $W$. It is well-known that $\mathbb{D}_{W}^{1,2}$ is a dense subset of $%
L^{2}(\Omega)$ and that $D$ is a closed and unbounded operator from $%
L^{2}(\Omega)$ into $L^{2}([0,T]\times\Omega).$ We will also consider the
iterated derivatives $D^{n},$ for $n>1,$ whose domains will be denoted by 
$\mathbb{D}_{Z}^{n,2}.$  We will also make use of the notation $\mathbb{L}^{n,2}:=L^{2}([0,T];\mathbb{D}_{Z}^{n,2}).$

\vspace{2mm}

Consider now a process $X$ given by the equation 
\begin{equation}
\label{model1}
dX_t=\phi^2_t dW_t-\frac12 \phi^2_tdt,
\end{equation}
where $\phi_t$ is a positive and square integrable process adapted to the filtration generated by $W$. We will make use of the following anticipating It\^o's formula (see for example Al\`os (2006)).

\begin{proposition}
\label{Ito}
Consider the model (\ref{model1}) and define the process $Y$ as $Y_t:=\int_t^T \phi^2_s ds$. Let
$F:[0,T]\times \mathbb{R}^{2}\rightarrow \mathbb{R}$ be a function
in $C^{1,2} ([0,T]\times \mathbb{R}^{2})$ such that there exists a
positive constant $C$ such that, for all $t\in \left[ 0,T\right]
,$ $F$ and its partial derivatives evaluated in $\left(
t,X_{t},Y_{t}\right)$ are bounded by $C.$ Then it follows that
\begin{eqnarray}
F(t,X_{t},Y_{t}) &=&F(0,X_{0},Y_{0})+\int_{0}^{t}{\partial _{s}F}%
(s,X_{s},Y_{s})ds \nonumber\\
&&-\frac12\int_{0}^{t}{\partial _{x}F}(s,X_{s},Y_{s})\phi^2_sds\nonumber\\
&&+\int_{0}^{t}{\partial _{x}F}(s,X_{s},Y_{s})\sqrt{\phi^2 _{s}} dW_{s} \nonumber\\
&&-\int_{0}^{t}{\partial _{y}F}(s,X_{s},Y_{s})\phi^2_{s}ds+
\int_{0}^{t}{\partial _{xy}^{2}F}(s,X_{s},Y_{s})\Theta _{s}ds \nonumber\\
&&+\frac{1}{2}\int_{0}^{t}{\partial
_{xx}^{2}F}(s,X_{s},Y_{s})\phi^2_{s}ds ,
\label{aito}
\end{eqnarray}
where $\Theta _{s}:=(\int_{s}^{T}D^W_{s}\phi_{r}^{2}dr)\phi _{s}.$ 
\end{proposition}

\section{Exotic options and vanilla options}
This paper focuses on the study of exotic options with maturity $T$, defined by a payoff of European type
$$
(A-K)_+,
$$
where $A$ is a square-integrable random variable defined on a risk neutral probability space $\left( \Omega ,\mathcal{F},\PP\right) $. We assume that $A$ is $\mathcal{F}^W_T$-adapted,  where $\mathcal{F}^W$ denotes the sigma-algebra generated by a $d$-dimensional Brownian motion $W$. Assume, for the sake of simplicity, that the interest rate is zero. Then, the price of this option at a moment $0<t<T$ is given by
$$
V_t:=E (A-K)_+,
$$
where $E$ denotes the expectation with respect to the probability $\PP$.\\\\
This section is devoted to see that the above exotic option can be seen as a European call option on a forward stock, under a stochastic volatility model. Towards this end, we define the martingale $M_t^T:=E_t (A)$. It is clear that
$$
V_t=E (M_T^T-K)_+.
$$
Observe that, by the martingale representation theorem (see, for instance, Karatzas and Shreve \cite{KS97},
theorem 3.4.15), there exists a d-dimensional $\mathcal{F}_t^W$-adapted process $(m_1(T, \cdot), . . . ,m_d(T, \cdot))$
such that
\begin{equation}
\label{martingale}
M_t^T=E(M_T^T)+\sum_{i=1}^d\int_0^t m_i (T,s)dW_s^i.  
\end{equation}
\begin{example}[Vanilla options]
Consider a vanilla call option with $A=F_T$, where $F$ is a forward stock price given by a stochastic volatility model of the form
\begin{equation}\label{eq:stock}
d F_{t} = \phi_{t}F_t \left( \rho dW_{t}^{1}+\sqrt{1-\rho ^{2}}W_{t}^{2}\right),
\end{equation}
where $\phi $ is a positive and square-integrable process adapted
to the filtration generated by $W^{1}$. Then (\ref{martingale}) holds with $d=2$, $M_t^T=F_t$, $m_1 (T,t)=\rho\phi_t F_t$ and $m_2 (T,s)=\sqrt{1-\rho ^{2}}\phi_t F_t$.
\end{example}
\begin{example}[VIX options]
\label{VIX}
Consider as underlying the random variable $VIX_T$  given by
\begin{equation}\label{eq:VIX}
VIX_T=\sqrt{\frac{1}{\Delta}E_T\int_T^{T+\Delta} v_s ds},
\end{equation}
where $v $ is a positive process adapted
to the filtration generated by a Brownian motion $W$. Then, $d=1$ and, if $VIX_T\in \mathbb{D}_W^{1,2}$, the Clark-Ocone formula and Proposition 1.2.8 in \cite{Nualart} give us that
\begin{eqnarray}
\label{mgeneral}
m(T,t)&=&\frac{1}{2\Delta} E_t \left(\frac{1}{VIX_T}E_T\int_T^{T+\Delta} (D_t v_s)ds\right)\nonumber\\
&=&\frac{1}{2\Delta} E_t \left(\frac{1}{VIX_T}\int_T^{T+\Delta} (D_t v_s)ds\right).
\end{eqnarray}
\end{example}
\begin{example}[Realized variance options]
\label{VarOp}
Consider the case
\begin{equation}\label{eq:Var}
RV_T=\frac1T\int_0^{T} v_s ds,
\end{equation}
where $v$ is a positive and square-integrable process adapted process as in Example \ref{VIX}. Then, $d=1$ and the Clark-Ocone formula gives us that (\ref{martingale}) holds with
$$
m(T,t)=\frac1T \int_t^{T}  E_t (D_t v_s)ds.
$$
\end{example}
\begin{remark}
\label{remark15}
We remark that, in the case $d=1$,  (\ref{martingale}) can be written as
\begin{equation}
\label{martingale2}
M_t^T=E(M_T^T)+\int_0^t \frac{m (T,s)}{M_t^T}M_t^T dW_s.  
\end{equation}
That is, a European option on $A$ can be seen as a European call option on a forward stock given by a stochastic volatility volatility model of the form (\ref{eq:stock}) with $\rho=1$ and
\begin{equation}
\label{phi}
\phi_t=\frac{m(T,t)}{M_t^T}.
\end{equation}
In particular, when $A=VIX_T$, 
\begin{equation}
\label{volVIX}
\phi_t=\frac{1}{2\Delta M_t^T} E_t \left(\frac{1}{VIX_T}\int_T^{T+\Delta} (D_t v_s)ds\right)
\end{equation}
and when $A$ is the RV,
\begin{equation}
\label{volvar}
\phi_t=\frac1T\frac {\int_t^{T} E_t (D_t v_s)ds}{ M_t^T}.
\end{equation}
\end{remark}

\section{Implied volatility for exotic options}
Our purpose in this section is to develop the tools to study the short-time behavior of the at-the-money implied volatility (ATMI) for exotic options. Later, we will apply these results to VIX options and RV options, where we assume the underlying volatility process given by $\sqrt{v}$, where $v$ is a postive and square-integrable process adapted to the filtration generated by a Brownian motion $W$.  
\\\\
Let us define the implied volatility as the quantity $I_t^T(k)$ such that
$$
V_t=BS(t,\ln(E(A)), k, I_t^T(k)),
$$
where $BS(t,x,k,\sigma)$ denotes the classical Black-Scholes price for a European call with time to maturity $T-t$, log-stock price $x$, log-strike price $k$ and volatility $\sigma$. That is, 

\[
BS(t,x,k,\sigma )=e^{x}N(d_{+}(k,\sigma ))-e^{k}N(d_{-}(k,\sigma )), 
\]%
where $N$ denotes the cumulative probability function of the standard normal
law and

\[
d_{\pm }\left( k,\sigma \right) :=\frac{x-k}{\sigma \sqrt{T-t}}\pm 
\frac{\sigma }{2}\sqrt{T-t}.
\]%
For sake of simplicity we will denote $I_t^T:=I_t^T(\ln(E(A))$ the corresponding ATMI. Notice that $I_t^T=BS^{-1}(t,\ln(E_t A),\ln(E_tA),v)$. In the sequel we will write $BS^{-1}(t,v)=BS^{-1}(t,\ln(E_t A),\ln(E_tA),v)$.

\vspace{2mm}

\subsection{The ATMI short-time limit}

Now our first objective is to study the short-time behaviour of the implied volatility level of exotic options. We will need the following hypotheses. 

\begin{description}
\item[(H1)] $A\in \mathbb{L}^{1,p}$ for all $p>1$ 
\item[(H2)] $\frac{1}{M^T}\in L^p$, for all $p>1$.

\item[(H3)] The terms
$$
E\left( \int_{0}^{T}\frac{1}{u_s^2(T-s)}\Theta _{s}ds\right), \label{preciohipotesis}
$$
and
$$
\frac{1}{T^2 u_0^T}E\left( \int_{0}^{T}\left(\int_s^T D_s\phi_r^2 dr\right)^2ds\right), \label{preciohipotesis2}
$$
are well defined and tend to zero as $T\to 0$, where $u_t^T:=\sqrt{\frac{1}{T-t}\int_t^T\phi_s^2ds}$, with $\phi$ and $\Theta$  defined as in (\ref{phi}) and Theorem \ref{Ito}.

\item[(H4)] There exists $\gamma\in\left(-\frac12,0\right]$ such that the term
$$
\frac{1}{T^{\frac12+\gamma}}E\sqrt{\int_0^T\phi_s^2ds}
$$ 
has a finite limit as $T\to 0$.
\end{description}

\begin{remark}\label{rem:Dsphi}
Notice that, if $A\in \mathbb{L}^{1,p}$, the process $M_t^T$ is also in $\mathbb{L}^{1,p}$ and $D_sM_t^T=E_t (D_s A)$ (see Proposition 1.2.8 in Nualart (2005)). Moreover, Clark-Ocone formula gives us that $m(T,t)=E_t(D_t A)$, which allows us to see that $m(T,t)\in\mathbb{L}^{1,2}$ and $D_s m(T,t)=E_t(D_s D_t A)$ for all $s<t$. Then it follows that $\phi$ is also in $\mathbb{L}^{1,p}$ and
\begin{eqnarray*}
\label{der}
D_s\phi_t&=&\frac{D_s m(T,t)M_t^T-m(T,t)D_sM_t^T}{(M_t^T)^2}\nonumber\\
&=&\frac{E_t(D_s D_t A) E_t(A)-E_t(D_tA)E_t(D_sA)}{(E_t(A))^2}.
\end{eqnarray*}

\end{remark}
We will need the following result, that is an adaptation of Theorem 4.2 in \cite{AlosLeonVives}to the context of our problem.

\begin{theorem}  [Adaptation of Theorem 4.2 in Al\`os, Le\'on and Vives (\cite{AlosLeonVives})]
\label{the:liv} 

Assume that hypotheses (H1), (H2) and (H3)  hold. Then we have that the option price $V_t=E(A_T-E(A_t))$ is given by
\begin{equation}
V_{t}=E_{t}\left( BS\left( 0,\ln(E(A)),k,u_{t}\right) \right) +\frac{1 }{2}%
E\left( \int_{0}^{T}\frac{\partial G}{\partial x}%
(s,\ln(M_{s}^T),\ln(E(A)),u_{s})\Theta _{s}ds\right), \label{precio}
\end{equation}
where $G:=(\partial^2_{xx}-\partial_{x})BS$.
\end{theorem}
\begin{proof}
This proof is based on the same arguments as in the proof of Theorem 4.2 in \cite{AlosLeonVives}. Notice that
$$ V_t=E_{t}\left( BS\left( T,\ln(A),\ln(E(A)),u_{T}\right) \right).$$
Then, applying Theorem \ref{Ito} and taking expectations we obtain that
\begin{eqnarray}
\label{decomp1}
V_t&=&E_{t}\left( BS\left( T,\ln(A),\ln(E(A)),u_{T}\right) \right)\nonumber\\
&=&E_{t}\left( BS\left( t,\ln(E(A),\ln(E(A)),u_{0}\right) \right)\nonumber\\
&+&E_t \int_{t}^{T}\mathcal{L}_{BS} (u_s)BS(s,\ln(E_s(A)),\ln(A),u_s)ds\nonumber\\
&+&E_{t}\left( \int_{t}^{T}\frac{\partial^2 BS}{\partial x\partial u}%
(s,\ln(E_s(A)),\ln(E(A)),u_{s})\frac{1}{2u_s\sqrt{T-s}}\Theta _{s}ds\right),
\end{eqnarray}
where $\mathcal{L}_{BS}(u_s)$ denotes the Black-Scholes operator $\mathcal{L}_{BS}(u_s)=\frac12\left(\partial^2_{xx}-\partial_{x}\right)u_s^2+\partial_s$.
Now the result follows from the fact that $\mathcal{L}_{BS} (u_s)BS(s,X_s,k,u_s)=0$ and taking into account that $\frac{\partial BS}{\partial u}=u_s\sqrt{T-s} G$. Notice that by (H3) and the fact that
\[
\frac{\partial G}{\partial x}(s,x,k,\sigma)=\frac{e^{x}N^{\prime }(d_{+}\left(
k,\sigma\right) )}{\sigma\sqrt{T-s}}\left( 1-\frac{d_{+}\left(k,\sigma\right) 
}{\sigma\sqrt{T-s}}\right),
\] the integral term in (\ref{decomp1}) is well defined.
\end{proof}\\\\
The following result is a direct consequence of the arguments  in the proof of Theorems 3 and 8 in Al\`{o}s and Shiraya \cite{AlosShiraya}.

\begin{theorem}
\label{ATMIlimit}
Assume that hypotheses (H1), (H2) and (H3) hold. Then 
$$
\lim_{T\to 0} \left(I_0^T-E\sqrt{\frac1T\int_0^T\phi_s^2ds}\right)=0.
$$
\end{theorem}
\begin{proof}
This proof is based on the same arguments as in Theorems 3 and 8 in \cite{AlosShiraya}. By definition, the ATMI is given by 
\begin{eqnarray}
I_0^T&=&BS^{-1}(0,\ln(E_t A),\ln(E(A)),V_0)\nonumber\\
&=&BS^{-1}(0,\ln(E_t A),\ln(E(A)),\Gamma_T),\nonumber
\end{eqnarray}
where
\begin{eqnarray*}
\Gamma_r&=&E\left( BS\left( 0,\ln(E(A)),k,u_{0}^T\right) \right) \nonumber\\
&+&\frac{1 }{2}%
E\left( \int_{0}^{r}\frac{\partial G}{\partial x}%
(s,\ln(M_{s}^T),\ln(E(A)),u_{s}^T)\Theta _{s}ds\right), \label{precio}
\end{eqnarray*}
Now, Theorem \ref{the:liv} allows us to write
\begin{eqnarray*}
\label{primerpaso}
&&I_0^T= E\left( BS^{-1}(0,\Gamma_0 \right))\nonumber\\
&+&E \int_t^T (BS^{-1})'(\ln(E(A)),\Gamma_s)\frac{\partial G}{\partial x}(s,\ln(E_sA),\ln(E(A)),u_{s}^T)\Theta _{s}ds,
\end{eqnarray*}
where $(BS^{-1})'$ denotes the first derivative of $BS^{-1}$ with respect to $\Gamma$. Now, as as%
\[
\left( BS^{-1}\right) ^{\prime }\left(k,\Gamma _{s}\right) =\frac{1%
}{e^{x}N^{\prime }(d_{+}\left( k,\Gamma _{s}\right)) \sqrt{T}},
\]%
and 
\[
\frac{\partial G}{\partial x}(s,x,k,\sigma)=\frac{e^{x}N^{\prime }(d_{+}\left(
k,\sigma\right) )}{\sigma\sqrt{T-s}}\left( 1-\frac{d_{+}\left(k,\sigma\right) 
}{\sigma\sqrt{T-s}}\right),
\] 
it is easy to see that (H3) implies that the second term in (\ref{primerpaso}) tends to zero. For the first term, we can write
\begin{eqnarray*}
\label{kenichiro}
&&\left( BS^{-1}(0,E( BS\left( 0,\ln(E(A)),\ln(E(A)),u_0^{T}) \right)\right)\nonumber\\
&=&E\left( BS^{-1}(0, BS\left( 0,\ln(E(A)),\ln(E(A)),u_0^{T}) \right)\right)\nonumber\\
&+&E\left( BS^{-1}(0,E( BS\left( 0,\ln(E(A)),\ln(E(A)),u_0^{T}) \right)\right)\nonumber\\
&-& BS^{-1}(0, BS\left( 0,\ln(E(A)),\ln(E(A)),u_0^{T}) \right)
\end{eqnarray*}
Notice that the first term in the right hand side of the above equation is equal to $u_0^T$. For the last two terms we can write
\begin{equation}
\label{mart}
BS\left( 0,\ln(E(A)),\ln(E(A)),u_0^{T}) \right) =BS\left( 0,\ln(E(A)),\ln(E(A)),u_0^{T}) \right)+\int_{t}^{T}U_{s}dW_{s},
\end{equation}
where, by the Clark-Ocone formula, 
\begin{eqnarray}
\label{U}
U_{s} &=&E_{s}\left[ D_{s}^{W}\left(BS\left( 0,\ln(E(A)),\ln(E(A)),u_0^{T}) \right)\right) \right]   \nonumber \\
&=&E_{s}\left[ AN^{\prime }(d_{+}\left(
\ln(E(A)), u_{0}^T\right))
\frac{\int_{s}^{T}D_{s}^{W}\phi _{s}^{2}ds}{2\sqrt{T}u_0^T}\right].
\end{eqnarray}
Define $\Lambda_r=E_r\left(BS\left( 0,\ln(E(A)),\ln(E(A)),u_0^{T}) \right)\right)$. Then, classical It\^o's formula and (\ref{mart}) give us that
\begin{eqnarray}&&E\left( BS^{-1}(0,E( BS\left( 0,\ln(E(A)),\ln(E(A)),u_0^{T}) \right)\right)\nonumber\\
&-& BS^{-1}(0, BS\left( 0,\ln(E(A)),\ln(E(A)),u_0^{T}) \right)\nonumber\\
&=& {E_{t}\left[ BS^{-1}(\ln(E(A)),\Lambda_0 )-BS^{-1}(\ln(E(A)),\Lambda_T )\right] }  \nonumber
\\
&=&-\frac12E_t\Bigg[\int_{t}^{T}\left( BS^{-1}\right) ^{\prime \prime}\left(\ln(E(A)), \Lambda_r\right) U_{r}^{2}dr\Bigg]. \label{zerocorr} 
\end{eqnarray}
This, jointly with (\ref{U}), (H3) and the fact that 
$$
\left( BS^{-1}\right) ^{\prime \prime }\left(X_0,{\Lambda _{r}}%
\right) =
\frac{BS^{-1}(X_0, \Lambda_r)}{4\left( \exp
(X_{t})N^{\prime }(d_{+}\left( X_0,BS^{-1}(X_0, \Lambda_r)\right)
)\right) ^{2}},
$$
allows us to prove that the last term in (\ref{zerocorr}) tends to zero. Now the proof is complete
\end{proof}
\\\\
This result gives us, if (H4) holds, the following corollary.

\begin{corollary}
\label{ATMIlimitcor}
Assume a random variable $A$ such that hypotheses (H1), (H2), (H3) and (H4) hold. Then 
$$
\lim_{T\to 0} T^{-\gamma} I_t^T= \lim_{T\to 0}\frac{1}{T^{\frac12+\gamma}}E\sqrt{\int_0^T\phi_s^2ds}.
$$
\end{corollary}

\subsection{The ATMI short-time skew}

Our main goal in this section is to study the ATM short-time limit of the implied volatility skew. In particular, we will characterize the 
class of stochastic volatilty process that reproduce a positive VIX skew, as observed in real market data.
Towards this end we will need the following hypotheses

\begin{description}

\item[(H1')] $A\in \mathbb{L}^{2,p}$ for all $p>1$ 
\item[(H5)] The terms
\begin{eqnarray*} \frac{1}{\sqrt{T}}\sum_{k=4}^6E\left(
\int_{0}^{T} (u_s^T (T-s))^{-\frac{k}{2}}\left(  \int_s^T\Theta_rdr\right)\Theta_{s}ds\right)
\end{eqnarray*}
and 
\begin{eqnarray*}\frac{1}{\sqrt{T}} \sum_{k=3}^4E\left(
\int_{0}^{T} \frac{1}{\sqrt{T}} (u_s^T (T-s))^{-\frac{k}{2}}\left(  \int_s^TD_s \Theta_rdr\right)\phi_{s}ds\right)
\end{eqnarray*}
tend to zero as $T\to 0$.

\item[(H6)]There exists $\lambda\in\left(-\frac12,0\right)$ such that the  term
$$
\frac{E \left[ \int_{0}^{T}\left( \phi_s^T D_{s}^W\phi_u^2 du\right) ds\right] }{(u _0^T)^{3}T^{2+\lambda
}}
$$
has a finite limit as $T\to 0$.

\end{description}

\begin{theorem}
\label{theadaptation} Consider a random variable $A$ such that
Hypotheses (H1'), (H2), (H3), (H5) and (H6) hold. Then%
\begin{equation}
\lim_{T\rightarrow t}T^{-\lambda}\frac{\partial I_0^T}{\partial k}%
(\ln(E(A))=\frac{1}{2}\lim_{T\rightarrow t}\frac{%
E_{t}\left[ \int_{0}^{T}\left( \phi_s \int_s^T D_{s} \phi_u^2 du\right) ds\right] }{u _0^{3}T^{2+\lambda
}}  \label{limitfirstderivative}
\end{equation}
\end{theorem}

\begin{proof}
This proof follows the same steps as those of Proposition 5.1 and
Proposition 6.2 in \cite{AlosLeonVives} and Theorem 4.5 in \cite{AL}. Taking partial derivatives with respect to $k$ on the expression $%
V_{0}=BS\left( 0,\ln(E(A)),k,I_0^T(k)\right) $ we obtain

\begin{eqnarray}
\frac{\partial V_{0}}{\partial k}&=&\frac{\partial BS}{\partial k}%
(0,\ln(E(A)),k,I_0^T(k)))\nonumber\\
&&+\frac{\partial BS}{\partial \sigma }(0,\ln(E(A)),k,I_0^T(k)))%
\frac{\partial I_0^T}{\partial k}(k).  \label{first}
\end{eqnarray}%
On the other hand, from (\ref{precio}) we deduce that 
\begin{equation}
\frac{\partial V_{0}}{\partial k}=E_{t}\left( \frac{\partial BS}{%
\partial k}(0,\ln(E(A)),k,u_0^T)\right) +E\left( \int_{0}^{T}\frac{\partial F%
}{\partial k}(s,\ln(E_s(A))),k,u_{s}^T)\Theta_{s}ds\right),  \label{second2}
\end{equation}%
 where 
$$
F(s,x,k,\sigma):=\frac{1}{2}\frac{\partial G}{\partial x}(s,x,k,\sigma).
$$
After some algebra (see \cite{AlosLeonVives}) it is easy to see that
\begin{eqnarray*}
&&E\left( \frac{\partial BS}{\partial k}(0,\ln(E(A)),\ln(E(A)),u_0^T)-%
\frac{\partial BS}{\partial k}(t,\ln(E(A)),\ln(E(A)),I_0^T))\right)  \\
&=&\frac{1 }{2}E\left( \int_{0}^{T} F(s,\ln(E_s(A))),k,u_{s}^T)\Theta_{s}ds\right) ,
\end{eqnarray*}%
This, jointly with (\ref{first}) and (\ref{second2}) implies that
\begin{equation}
\label{laquefaltaba}  
\frac{\partial I_0^T}{\partial k}(\ln(E(A))  
 =\frac{E\left(
\int_{0}^{T}L(s,\ln(E_s(A))),k,u_{s}^T)\Theta_{s}ds\right)  }{\frac{\partial BS}{\partial \sigma }\left(
0,\ln(E(A)), \ln(E(A)),I_0^T\right) },
\end{equation}%
where $L:=(\frac12+\frac{\partial}{\partial k})F$. Now
\begin{eqnarray}
\label{T}
&&E\left(
\int_{0}^{T}L(s,\ln(E_s(A))),k,u_{s}^T)\Theta_{s}ds\right)  \nonumber \\
&=&E \left(L(0,\ln(E(A)),k,u_0^T)\int_{t}^{T}\Theta _{s}ds\right)\nonumber \\
&+&\frac12 E\left(
\int_{0}^{T}\left(\frac{\partial^3}{\partial x^3}-\frac{\partial^2}{\partial x^2}\right)L(s,\ln(E_s(A))),k,u_{s}^T)\left(  \int_s^T\Theta_rdr\right)\Theta_{s}ds\right)\nonumber\\
&+& E\left(
\int_{0}^{T}\frac{\partial L}{\partial k}(s,\ln(E_s(A))),k,u_{s}^T)\left(  \int_s^TD_s\Theta_rdr\right)\phi_s ds\right)\nonumber\\
&=:&T_1+T_2+T_3.
\end{eqnarray}
Now, a simple calculation gives us that 
\begin{equation}
\label{T2}
\left(\frac{\partial^3}{\partial x^3}-\frac{\partial^2}{\partial x^2}\right)L(s,\ln(E_s(A))),k,u_{s}^T)\leq
\sum_{k=4}^6 (u_s^T (T-s))^{-\frac{k}{2}},
\end{equation}
\begin{equation}
\label{T3}
\frac{\partial L}{\partial k}(s,\ln(E_s(A))),k,u_{s}^T)\leq
\sum_{k=3}^4 (u_s^T (T-s))^{-\frac{k}{2}},
\end{equation}
and
\begin{equation}
\label{T1}
L(0,\ln(E(A)),k,u_0^T)= \frac{1}{2}\frac{A\exp \left( -\frac{(u_{0}^T)^{2}T}{8}%
\right) }{\sqrt{2\pi }u_{0}^T\sqrt{T}}\left[ \frac{1}{(u_{0}^T)^{2}T}-\frac{1%
}{2}\right] .
\end{equation}%
Then, (\ref{T2}), (\ref{T3}), (\ref{T1}) together with (\ref{laquefaltaba}), (\ref{T}) and the fact that 
$$\frac{\partial BS}{\partial \sigma }\left(0,x,x,\sigma\right)= \frac{A\exp \left( -\frac{\sigma^{2}T}{8}%
\right)\sqrt{T} }{\sqrt{2\pi }}
$$
allow us to complete the proof.
\end{proof}

\begin{remark}
\label{remark11}Notice that $D_s\phi_t^2=2\phi_tD_s\phi_t$. Then the above theorem gives us that the sign of the short-time skew is positive if $$\lim_{s,t\to 0}E\left(f(s,t)D_s\phi_t- C\right)^p=0$$  
for some a positive function $f(s,t)$ such that $\lim_{s,t\to 0}f(s,t)\geq 0$ and for some constants $C>0$ and $p>1$. By (\ref{phi}), this condition holds if
\begin{eqnarray}
\label{criteria}
\lim_{s,t\to 0}E\left[f(s,t)(D_s m(T,t)M_t^T-m(T,t)D_sM_t^T)- c\right]^q=0,
\end{eqnarray}
for some $c>0$ and $q>1$.
Notice that (\ref{criteria}) gives us a tool to identify those stochastic volatility models that can reproduce a short-time positive skew, as we will see in the next section.
\end{remark}
\section{A General family of models}
Building on the Truncated Brownian semistationary processes ($\mathcal{TBSS}$) introduced by Barndorff-Nielsen and Schmiegel \cite{turbulence} we consider the following family of models: 
\begin{model}[Mixed Generalized rough volatility models]\label{MixedTBSS}
Let us define the following instantaneous variance dynamics:
$$v_t=v_0 \left(\gamma\mathcal{E}\left(\nu\sqrt{2H}\mathcal{B}_t\right)+(1-\gamma)\mathcal{E}\left(\eta\sqrt{2H}\mathcal{B}_t\right)\right),\quad \nu\geq 0,\;\eta\geq 0$$
where $$\mathcal{B}_t=\int_0^t\exp\left(-\beta(t-s)\right) (t-s)^{H-1/2} dW_s,\quad \beta \geq 0$$ and $\mathcal{E}(\cdot)$ denotes the Wick stochastic exponential.  It is then easy to see that $$D_sv_u=\left(\gamma \nu\mathcal{E}(\nu\sqrt{2H}\mathcal{B}_u)+(1-\gamma)\eta\mathcal{E}(\eta\sqrt{2H}\mathcal{B}_u\right)\sqrt{2H} (t-s)^{H-1/2} \exp(-\beta(u-s)),$$ for all $s<u$. Then  hypotheses (H1)-(H6) hold for the underlyings $VIX_T$ and $RV$, with $\lambda=0$ and $\lambda=H-\frac12$, respectively.
\end{model}
This class of models covers most of the lognormal models considered in the literature (SABR \cite{SABR}, (rough) Bergomi \cite{Bergomi,BFG15}, etc.). Moreover, we consider a mixing weight solution as proposed by Bergomi \cite{Bergomi3} to overcome flat VIX smiles. We also emphasize that \cite{HJM17,JMM17} provide theoretical justification to the Monte Carlo simulation of these processes.
\begin{remark}
The term mixing makes reference to the fact that we have a sum of two lognormal random variables with different weights. The reader should note that when we do not explicitly use the term mixed, this corresponds to $\gamma=\eta=0$.
\end{remark}
\section{VIX options}

This section is devoted to apply the results in Section 4 to the study of the ATMI level and skew of VIX options. First of all, let us illustrate the typical behaviour of these quantities in the market. Figure \ref{fig:VIXATMIHistorical} shows the historical behaviour of the ATMI term structure (source: OptionMetrics). A very sharp decrease of the ATMI level between the first two available maturities is usually observed. Also, the cloud of points in Figure \ref{fig:VIXATMIHistorical} is rather wide, meaning that volatility of volatility itself changes with time. In Figure \ref{fig:VIXSmiles} we observe that the skew is clearly positive, but also decreasing as a function of time to maturity, since smiles tend to flatten.
\begin{figure}[H]
\includegraphics[scale=0.8]{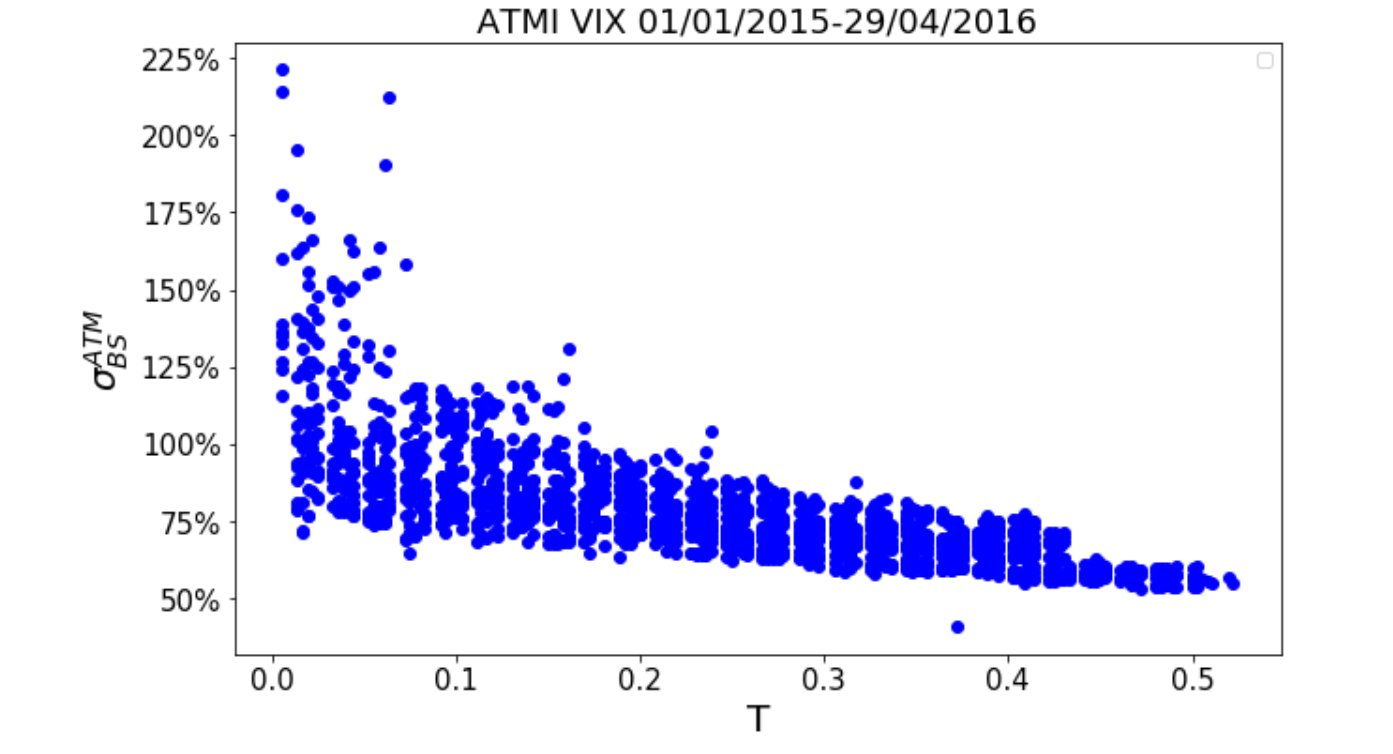}
\caption{VIX ATMI daily behaviour 2015-2016. Source: OptionMetrics.}
\label{fig:VIXATMIHistorical}
\end{figure}
\begin{figure}[H]
\includegraphics[scale=0.8]{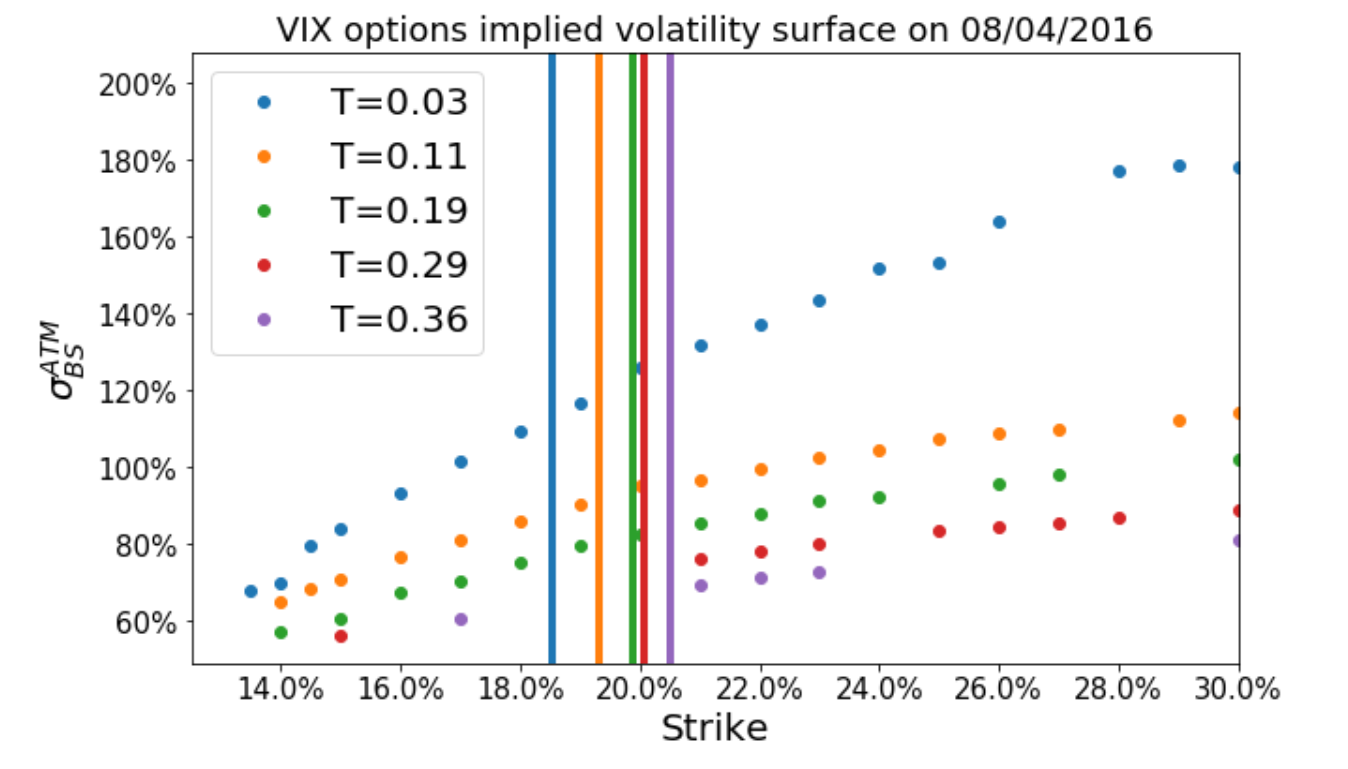}
\caption{VIX options implied volatility surface on 08/04/2016. Solid lines represent the ATMI level corresponding to each maturity. Source: OptionMetrics.}
\label{fig:VIXSmiles}
\end{figure}
\subsection{ATMI of VIX options}
\begin{proposition}\label{PROP: ATMIVIX}
Consider the process $VIX_T$ given by (\ref{eq:VIX}), with $v=f(Y)$,   where $f$ is a  function in $\mathcal{C}^2$ such that $f, f',f''\in L^p$, for all $p>1$, and $Y=\int_0^t (t-s)^{H-\frac12}g(t-s) dW_s$, where $H<\frac12$ and $g\in\mathcal{C}^1_b$ . Then we have
\begin{eqnarray}
\lim_{T\to 0} I_t^T(\ln E(VIX_T))&=\displaystyle \frac{f'(Y_0)}{2\Delta(VIX_0)^2}\phi(\Delta)\label{eq: generalformATMIVIX}.
\end{eqnarray}
\end{proposition}
\begin{proof}
 Then (H1)-(H4) hold with $\lambda=0$ and Corollary \ref{ATMIlimitcor} gives us that
 \begin{eqnarray*}
&&\lim_{T\to 0} I_t^T(\ln E(VIX_T))= \lim_{T\to 0}\frac{1}{T^{\frac12}}E\sqrt{\int_0^T\sigma_s^2ds}\nonumber\\
&&=\lim_{T\to 0}\frac{1}{T^{\frac12}}E\sqrt{\int_0^T\left(\frac{1}{2\Delta M_s^TVIX_s} \int_T^{T+\Delta}E_s (D_s(v_u))du\right)^2ds}\nonumber\\
&&=\frac{f'(Y_0)}{2\Delta M_0^0VIX_0}\left( \int_0^{\Delta}u^{H-\frac12}g(u)du\right)\nonumber\\
&&=\frac{f'(Y_0)}{2\Delta (VIX_0)^2}\left( \int_0^{\Delta}u^{H-\frac12}g(u)du\right)\nonumber\\
&&=:\frac{f'(Y_0)}{2\Delta(VIX_0)^2}\phi(\Delta).
\end{eqnarray*}
\end{proof}

\begin{example}[Mixed Generalized rough volatility models]
 We consider Model \ref{MixedTBSS}, then applying Proposition \ref{PROP: ATMIVIX} we get
\begin{eqnarray}
\label{eq:ATMIVIX}
&&\lim_{T\to 0}I_t^T(\ln E(VIX_T)\nonumber\\
&&=\displaystyle\frac{(\gamma \nu+(1-\gamma)\eta)\sqrt{2H}}{2\Delta}\left( \int_0^{\Delta}u^{H-\frac12}\exp(-\beta u)du\right)\nonumber\\
&=&\frac{(\gamma \nu+(1-\gamma)\eta)\sqrt{2H}}{2\Delta} \beta^{-H-\frac12}\Gamma_{low}\left(H+\frac12,\beta\Delta\right),
\end{eqnarray}
where $\Gamma(\cdot)$ is the Gamma function and $\Gamma_{low}(1+\alpha,x)=\int_0^xt^{-\alpha}e^{-t}dt$ is the lower incomplete Gamma function.\end{example}
In Figure \ref{fig:VIXATMI} we present numerical approximations with $\gamma=\eta=0$ that support the theoretical short time limit. In order to produce the Monte Carlo estimate, Algorithm 3.9 in Jacquier, Martini and Muguruza \cite{JMM17} and $T=10^{-4}$ was used. 

\begin{figure}[H]
\includegraphics[scale=0.8]{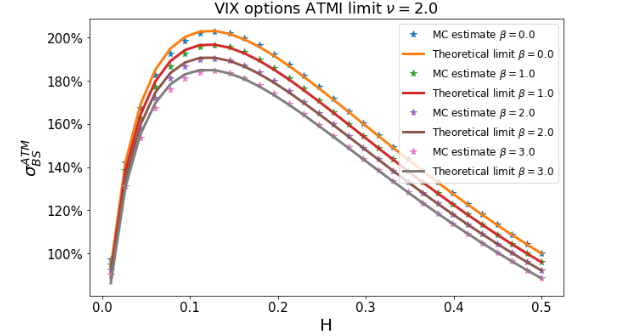}
\caption{VIX ATMI Short time limit in an exponential $\mathcal{TBSS}$ with $\nu=2$}
\label{fig:VIXATMI}
\end{figure}

\begin{remark}
If we consider the SABR case, i.e. $H=\frac{1}{2}$ and $\beta=0$ in \eqref{eq:ATMIVIX}, we obtain $\lim_{T\to 0}I_t^T=\frac{\nu}{2}$. In particular, we notice that this limit is not affected by the window size $\Delta$.
\end{remark}
\subsubsection{A semi close-form formula for the ATMI level of VIX}
We consider a small $T$ approximation of the limit given in \eqref{eq: generalformATMIVIX}, which turns out to be very sharp even for relatively long maturities (up to 5 years in some cases). We consider the following approximation inspired by  \eqref{eq: generalformATMIVIX}:
\begin{equation}\label{eq: close-form VIX ATMI}
I_T\approx \frac{f'(Y_0)}{v_0 2\Delta T^{\frac{1}{2}}}\sqrt{\int_0^T\left(\int_T^{T+\Delta}(u-s)^{H-\frac{1}{2}}g(u-s) du\right)^2 ds}
\end{equation}
In Figures \ref{fig:VIXATMI_formula} and \ref{fig:VIXATMI_formula2} we use formula  \eqref{eq: close-form VIX ATMI} with Model \ref{MixedTBSS}. As previously mentioned we observe that formula \ref{eq: close-form VIX ATMI} indeed performs very well, specially in the case $\gamma=\eta=0$ even for relatively large maturities.
\begin{remark} In practice, one may use approximation \eqref{eq: close-form VIX ATMI} to control the ATMI level of VIX in our model. Moreover, in most cases (SABR, exponential OU, exponential fBm, etc.)  \eqref{eq: close-form VIX ATMI} admits a close-form formula without any numerical integration. 
\end{remark}

\begin{figure}[H]
\includegraphics[scale=0.6]{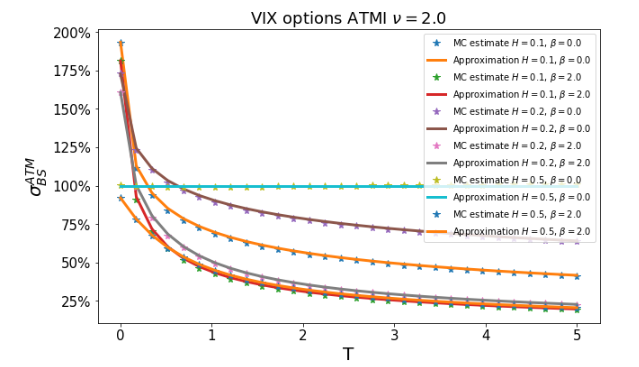}
\caption{VIX ATMI in a Generalized rough volatility model with $\nu=2$}
\label{fig:VIXATMI_formula}
\end{figure}
\begin{figure}[H]
\includegraphics[scale=0.6]{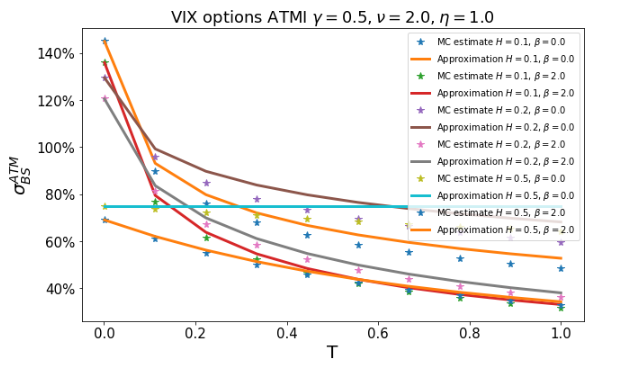}
\caption{VIX ATMI in a Mixed Generalized rough volatility model with $(\gamma,\nu,\eta)=(1/2,2,1)$}
\label{fig:VIXATMI_formula2}
\end{figure}

\subsection{ATMI skew of VIX options}
The following result proves that, for models based on $\mathcal{TBSS}$, the ATMI VIX skew if of the order $O(1)$ as time to maturity tends to zero.
\begin{proposition}\label{prop:fbmATMISkewVIX}
Consider an instantaneous variance model of the form 
$$v_t=v_0f(Y_t)$$ where $Y_t=\int_0^t (t-s)^{H-1/2} \exp(-\beta(t-s))dW_s$ for some $H\leq \frac12$ and some $\beta\geq0$, and where $f$ is a  function in $\mathcal{C}^2$ such that $f, f',f''\in L^p$, for all $p>1$. Then, the ATMI skew for VIX options is given by:
\begin{eqnarray*}
\displaystyle\lim_{T\rightarrow 0}\frac{\partial I_0^T}{\partial k}
(\ln E(VIX_T))=\frac{1}{2}\left(\frac{G(H,\Delta,\beta)}{J(H,\Delta,\beta)}\frac{f''(Y_0)}{f'(Y_0)}-\frac{J(H,\Delta,\beta)}{\Delta}\frac{f'(Y_0)}{(M_0^0) VIX_0 }\right).
\end{eqnarray*}
where
\begin{eqnarray*}G(H,\Delta,\beta)&&=\begin{cases} (2\beta)^{-2H}\Gamma_{low}(2H,2\beta\Delta) &\text{ if } \beta>0 \\
\displaystyle \frac{\Delta^{2H}}{2H} & \text{ if } \beta=0

\end{cases}\\
J(H,\Delta,\beta)&&=\begin{cases} \beta^{-H-1/2}\Gamma_{low}(H+1/2, \beta\Delta) &\text{ if } \beta>0 \\
\displaystyle \frac{\Delta^{H+1/2}}{H+1/2} & \text{ if } \beta=0

\end{cases}.
\end{eqnarray*}

\end{proposition}
\begin{proof}
The proof is postponed to Appendix \ref{sec:skewVIXOptions} to ease the flow of the paper.
\end{proof}
\begin{example}[Mixed Generalized rough volatility model]
 We consider Model \ref{MixedTBSS}, then we have that the ATMI skew for RV options is given by
\begin{eqnarray*}
&&\lim_{T\rightarrow 0}\frac{\partial I_0^T}{\partial k}
(\ln E(VIX_T))\nonumber\\
&&=\frac{\sqrt{2H}}{2}\left(\frac{\left(\gamma\nu^2+(1-\gamma)\eta ^2\right)G(H,\Delta,\beta)}{\left(\gamma\nu+(1-\gamma)\eta\right)J(H,\Delta,\beta)}-\frac{\left(\gamma\nu+(1-\gamma)\eta\right)J(H,\Delta,\beta)}{\Delta}\right).
\end{eqnarray*}

\end{example}
Figure \ref{fig:ATMIKsewLimitVIX} shows the accuracy of the asymptotic limit with a Monte Carlo benchmark of $T=10^{-4}$ and $\gamma=\eta=0$. The derivative was approximated using a central difference scheme.
\begin{figure}[H]
\includegraphics[scale=0.6]{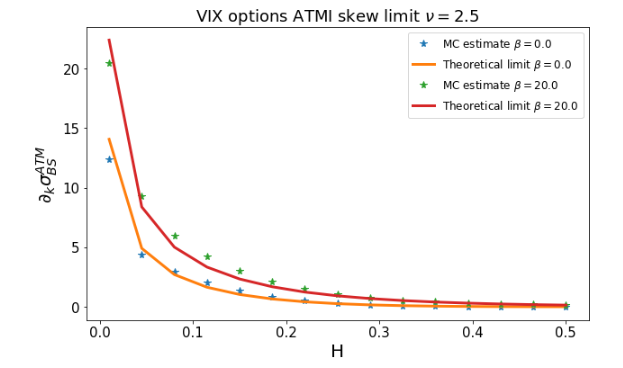}
\caption{ Generalized rough volatility model ATMI skew on VIX options with $\nu=2.5$}
\label{fig:ATMIKsewLimitVIX}
\end{figure}
\subsubsection{The sign of the VIX skew}
As indicated in Remark \ref{remark11}, the positivity of the short-time VIX skew is linked  to the positivity of
\begin{eqnarray}
D_s m(T,t)M_t^T-m(T,t)D_sM_t^T.\nonumber
\end{eqnarray}
In this section we will apply this criteria to study the short-time VIX skew corresponding to the Heston model and the SABR model. 
\begin{example}[Heston]
\label{Heston}
For the Heston model, we have that
$$
dv_t=k(\theta-v_t)+\nu\sqrt{v_t}dW_t,
$$
for some positive constants $k, \theta$ and $\nu$. Notice that, if the Feller condition $2k\theta>\nu^2$ holds, the Heston process is positive. For the sake of simplicity, we will assume $v_0=\theta$. Then we have
\begin{eqnarray}
\label{skew_heston}
&&D_s m(T,t)M_t^T-m(T,t)D_sM_t^T\nonumber\\
&&\to \frac{\nu^2(1-\exp(-k\Delta))}{4k\Delta}\left(1-\frac{2(1-\exp(-k\Delta))}{k\Delta}   \right),
\end{eqnarray}
in $L^2(\Omega)$, as $T\to 0$.
\end{example} 
\begin{proof}The computation of the limit is given in Appendix \ref{VIXSkewHeston}.
\end{proof}\\\\
Under reasonable market conditions for the reversion level $k$, we have that (\ref{skew_heston}) is negative. Then, condition (\ref{criteria}) holds with $f(t,s)=1$, which implies that the corresponding VIX skew is negative.
\begin{example}[SABR]
Consider the SABR model given by $v=\sigma^2$, where
$$
d\sigma_t=\alpha\sigma_tdW_t,
$$
for some positive constant $\alpha$. Then it follows that, for $s<t$, $D_sv_t=2\alpha v_t$. This implies that
$$
\lim_{t\to 0} M_t^T=v_0, \lim_{s,t\to 0} D_sM_t^T=2\alpha v_0, \lim_{s,t\to 0} D_sm_t^T=4\alpha^2v_0,
$$
from where we deduce that 
$$
D_s m(T,t)M_t^T-m(T,t)D_sM_t^T\to 0.
$$
This gives us that the SABR model generates a short-time flat skew.
\end{example}
\subsubsection{Lognormal models and flat VIX smiles}
Based on Figure \ref{fig:ATMIKsewLimitVIX} one may conclude that for small values of $H$ one can achieve large skews. Nevertheless we show in Figure \ref{fig:VIXSkewDecay} the the skew decays almost immediately. This behaviour is consistent with the fact that lognormal models generate flat smiles, hence the skew should be zero. On the contrary, Figure \ref{fig:VIXSkewDecay2} shows the effect of Mixing lognormals, which produces a much higher level of skew. In particular we notice that in the mixed SABR case ($H=1/2$ and $\beta=0$), the skew is given by
\begin{eqnarray}\label{SABRVIXSkew}
&&\lim_{T\rightarrow 0}\frac{\partial I_0^T}{\partial k}
(\ln E(VIX^{SABR}_T))\nonumber=\frac{\sqrt{2H}}{2}\left(\frac{\left(\gamma\nu^2+(1-\gamma)\eta ^2\right)}{\left(\gamma\nu+(1-\gamma)\eta\right)}-\left(\gamma\nu+(1-\gamma)\eta\right)\right).
\end{eqnarray}
We clearly have that \eqref{SABRVIXSkew} is non-zero unless $\gamma\in\{0,1\}$ or $\nu=\eta$. Moreover, in Figure  \ref{fig:VIXSkewDecay2} the skew is constant at this level for maturities up to 6 month. 
In order to obtain the approximating formulas we are inspired by the computations in Appendix \ref{sec:skewVIXOptions}, which yield the following expression before taking the limit:
\begin{eqnarray*}
&&\frac{\partial I_0^T}{\partial k}
(\ln E(VIX_T))\\&&\approx\frac{\sqrt{2H}}{\sqrt{T}\sqrt{\int_0^TK^2(T,\Delta,u)du}}\nonumber\\
&&\times\left(\frac{\left(\gamma\nu^2+(1-\gamma)\eta ^2\right)\int_0^{T}K(T,\Delta,s)\int_s^T K(T,\Delta,u)I(\Delta,T,s,u) du ds}{\left(\gamma\nu+(1-\gamma)\eta\right)}\right.\nonumber\\
&&\left.\hspace{0.3cm}-\frac{\left(\gamma\nu+(1-\gamma)\eta\right)\int_0^{T}K^2(T,\Delta,s)\int_s^T K^2(T,\Delta,u)du ds}{\Delta}\right).
\end{eqnarray*}
where $K(\cdot)$ and $I(\cdot)$ are defined in Appendix  \ref{sec:skewVIXOptions}. We observe that the approximation performs very well for maturities smaller than 6 months and could be useful to control the ATMI skew level of VIX in a given model.

\begin{figure}[H]
\includegraphics[scale=0.7]{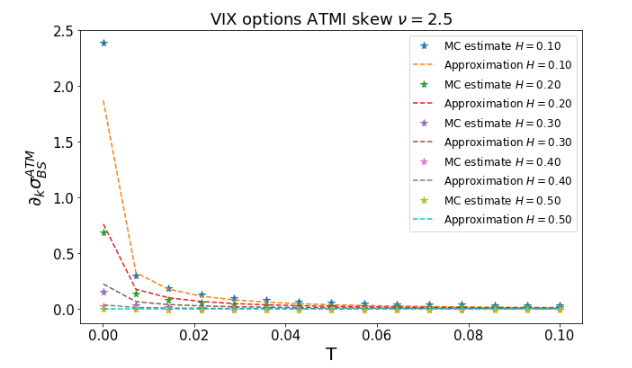}
\caption{ Generalized rough volatility model ATMI skew with $\nu=2.5$}
\label{fig:VIXSkewDecay}
\end{figure}

\begin{figure}[H]
\includegraphics[scale=0.6]{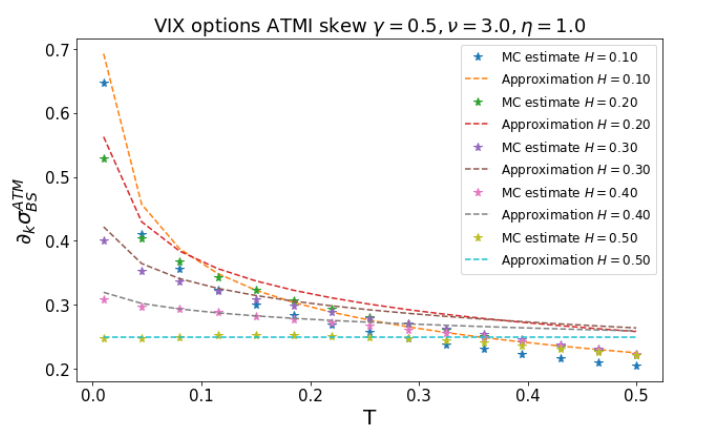}
\caption{Mixed generalized rough volatility model ATMI skew with $(\gamma,\nu,\eta)=(1/2,3,1)$}
\label{fig:VIXSkewDecay2}
\end{figure}

\section{Realized variance options}
This section is devoted to study the ATMI short-time level and skew of RV options. Figure \ref{fig:marketATMI} shows a empirical term structure of the ATMI, which is usually close to a power law.
\begin{figure}[H]
\includegraphics[scale=0.6]{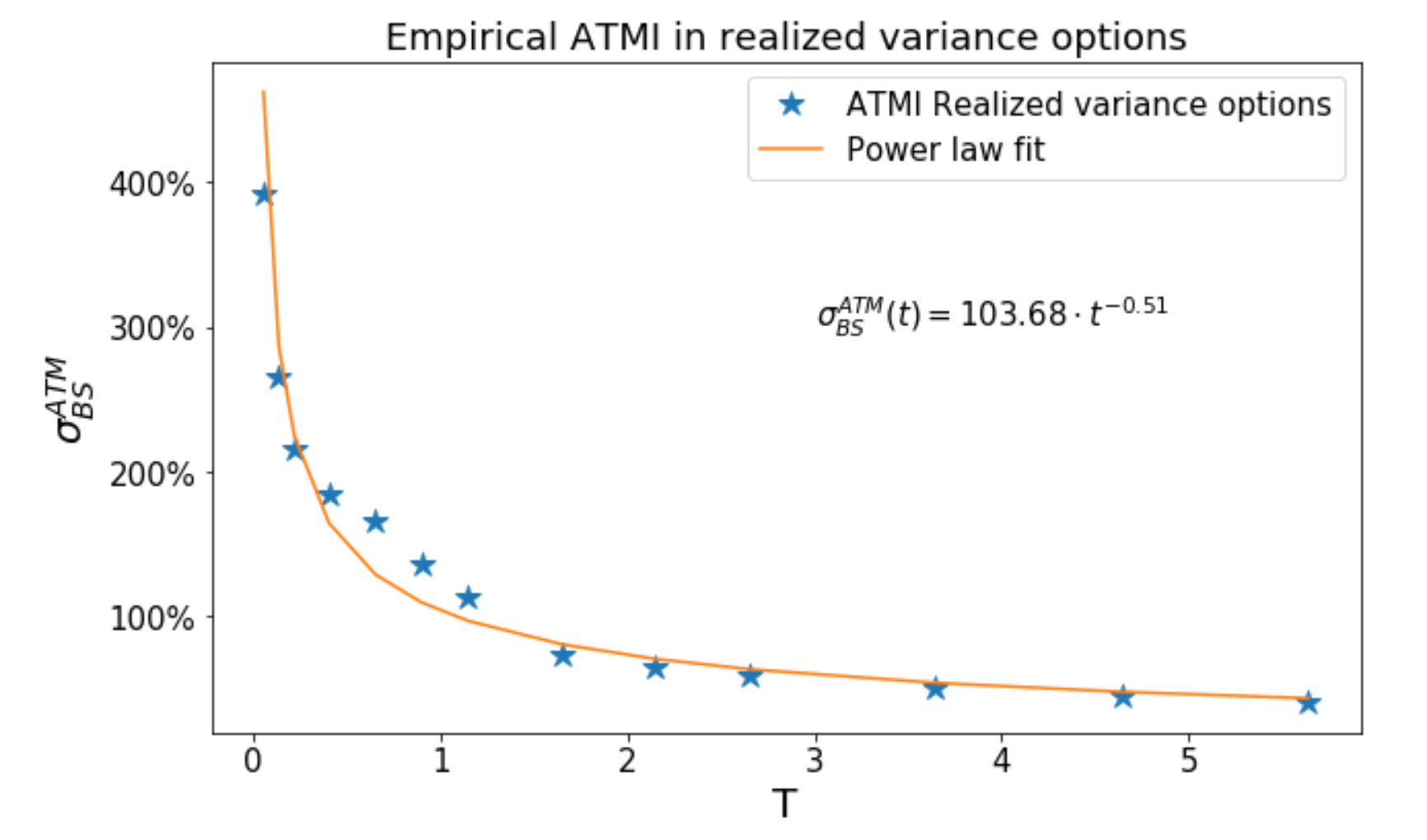}
\caption{Realized variance ATMI for S\&P 500 on February 27, 2018.}
\label{fig:marketATMI}
\end{figure}
\subsection{ATMI on realized variance options}
\begin{proposition}\label{PROP: ATMIRV}
Consider the process $v=f(Y)$,   where $f$ is a  function in $\mathcal{C}^2$ such that $f, f',f''\in L^p$, for all $p>1$, and $Y=\int_0^t (t-s)^{H-\frac12}g(t-s) dW_s$, where $H<\frac12$ and $g\in\mathcal{C}^1_b$ . Then we have that the short-time limit of the ATMI for RV options is given by
\begin{eqnarray}
\lim_{T\to 0} T^{\frac12-H}I_t^T(\ln E(RV_T))&=\displaystyle \frac{f'(Y_0)g(0)}{(H+\frac12)\sqrt{2H+2}M_0^0}\label{ATMIvar}.
\end{eqnarray}
\end{proposition}
\begin{proof}
We may directly apply Theorem \ref{ATMIlimit} since (H1), (H2) and (H3) hold with $\lambda=H-\frac12$. Thus,
\begin{eqnarray}
\label{ATMIvar}
&&\lim_{T\to 0}T^{\frac12-H} I_t^T(\ln E(RV_T))\nonumber\\
&&= \lim_{T\to 0}\frac{1}{T^{H}}E\sqrt{\int_0^T\phi^2_sds}\nonumber\\
&&=\lim_{T\to 0}\frac{1}{T^{1+H}}E\sqrt{\int_0^T\left(\frac{1}{M_s^T} \int_s^{T}E_s (D_s(v_u))du\right)^2ds}\nonumber\\
&&=\frac{f'(Y_0)g(0)}{M_0^0}\lim_{T\to 0}\frac{1}{T^{1+H}}E\sqrt{\int_0^T\left(\int_s^{T}(u-s)^{H-\frac12}du\right)^2ds}\nonumber\\
&&=\frac{f'(Y_0)g(0)}{(H+\frac12)M_0^0}\lim_{T\to 0}\frac{1}{T^{1+H}}E\sqrt{\int_0^T\left(T-s\right)^{2H+1}ds}\nonumber\\
&&=\frac{f'(Y_0)g(0)}{(H+\frac12)\sqrt{2H+2}M_0^0}.
\end{eqnarray}
\end{proof}
\begin{remark}
We emphasize that (\ref{ATMIvar}) implies that, for short maturities the ATMI is of order $O(T^{H-\frac12})$, which is consistent with real market data on equity markets as shown in Figure \ref{fig:marketATMI}.
\end{remark}

\begin{example}[Mixed Generalized Rough volatility]
\label{eq:ATMIExtendedRoughVl}
We consider again the Model \ref{MixedTBSS}, applying \ref{der} we obtain the following short-time ATMI limit:
\begin{eqnarray*}
&&\lim_{T\to 0}T^{1/2-H} I_t^T(\ln E(RV_T))\nonumber\\
&&=(\gamma\nu+(1-\gamma)\eta)\sqrt{2H}\lim_{T\to 0}\frac{1}{T^{H+1}}\sqrt{\int_0^T\left(\int_s^T\exp\left(-\beta(u-s)\right) (u-s)^{H-1/2} du\right)^2ds}\nonumber\\
&&=\frac{(\gamma\nu+(1-\gamma)\eta)\sqrt{2H}}{(H+\frac12)\sqrt{2H+2}}.
\end{eqnarray*}
\end{example}
\subsubsection{A semi close-form formula for the ATMI level of RV options}
Once again, we consider a short-time approximation of the ATMI level, inspired by the short-time limit given in Corollary \ref{ATMIvar}. More precisely, we consider the following approximation:
\begin{equation}\label{eq: close-form RV ATMI}
I_T\approx \frac{f'(Y_0)}{v_0 T^{\frac{3}{2}}}\sqrt{\int_0^T\left(\int_s^{T}(u-s)^{H-\frac{1}{2}}g(u-s) du\right)^2 ds}
\end{equation}
In Figure \ref{fig:VarianceATMI} we present the numerical results for the approximation for different values of $H$ and $\beta$ in Model \ref{MixedTBSS} with $\gamma=\eta=0$. We observe a very sharp fit to the Monte Carlo estimates obtained using the rDonsker scheme introduced in Horvath, Jacquier and Muguruza \cite{HJM17}.
\begin{remark}
In light of the accuracy of the approximating scheme given in \eqref{eq: close-form RV ATMI}, this can be used in practice to both control the ATMI level in a model and calibrate the parameters $H$ and $\beta$. In particular, when $\beta=0$ this reduces to a power-law fit of order $H-\frac{1}{2}$ as previously shown in Example \ref{eq:ATMIExtendedRoughVl}.
\end{remark}
\begin{figure}[H]
\hspace*{-1.5cm}\includegraphics[scale=0.8]{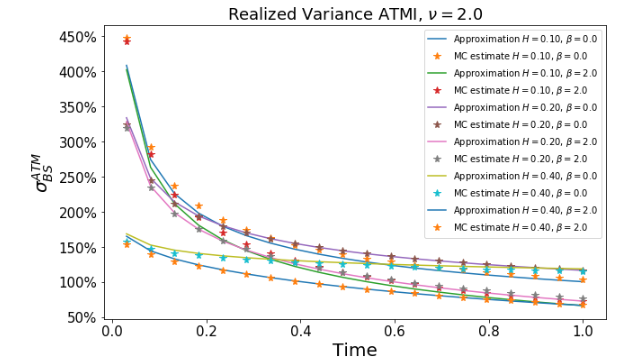}
\caption{Realized variance option ATMI with $\nu=2$}
\label{fig:VarianceATMI}
\end{figure}
\subsection{ATMI skew for realized variance options}
\begin{proposition}\label{prop:fbmATMISkewVariance}
Consider an instantaneous variance model of the form 
$$v_t=v_0f(Y_t)$$ where $Y_t=\int_0^t (t-s)^{H-1/2} dW_s$ and $f$ is a  function in $\mathcal{C}^2$ such that $f, f',f''\in L^p$, for all $p>1$, and $Y=\int_0^t (t-s)^{H-\frac12}g(t-s) dW_s$, where $H<\frac12$ and $g\in\mathcal{C}^1_b$
 Then, the ATMI skew for RV options is given by
\begin{eqnarray*}
&&\lim_{T\rightarrow 0}T^{\frac12-H}\frac{\partial I_0^T}{\partial k}
(\ln E(RV_T))\\
&&= \displaystyle\displaystyle \left(\frac{f''(Y_0)\mathcal{I}(H)(2H+2)^{3/2}(H+1/2)}{f'(Y_0)}-\frac{f'(Y_0)}{v_0(2H+1)\sqrt{(2H+2)}}\right).
\end{eqnarray*}
\end{proposition}
\begin{proof}
The proof and definition of $\mathcal{I}(H)$ are postponed to Appendix \ref{sec:skewVarianceOptions} to ease the flow of the paper.
\end{proof}
\begin{example}[Mixed Rough Bergomi]
We consider the Mixed rough Bergomi model as a subclass of the models presented in Model \ref{MixedTBSS} with instantaneous variance given by:
$$ v_{t} = v_{0}\gamma \mathcal{E}\left( \nu\sqrt{2H}\int_{0}^{t} (t-s)^{H-1/2} dW_{s} \right) +(1-\gamma)\mathcal{E}\left( \eta\sqrt{2H}\int_{0}^{t} (t-s)^{H-1/2} dW_{s} \right),$$ Applying Proposition \ref{prop:fbmATMISkewVariance} we have that the ATMI skew for RV options is given by
\begin{eqnarray*}
&&\displaystyle\lim_{T\rightarrow 0}\displaystyle T^{\frac12-H}\frac{\partial I_0^T}{\partial k}
(\ln E(RV_T))\nonumber\\
&&=\displaystyle \sqrt{2H}\left(\frac{\gamma\nu^2+(1-\gamma)\eta^2}{\gamma\nu+(1-\gamma)\eta}\mathcal{I}(H)(2H+2)^{3/2}(H+1/2)-\frac{\gamma\nu+(1-\gamma)\eta}{(2H+1)\sqrt{(2H+2)}}\right).
\end{eqnarray*}
 Figures \ref{fig:ATMIVarianceSkew} and \ref{fig:ATMIVarianceSkewMixed} illustrate the accuracy of the asymptotic formula for different values of $H$.\end{example}
\begin{figure}[H]
\includegraphics[scale=0.7]{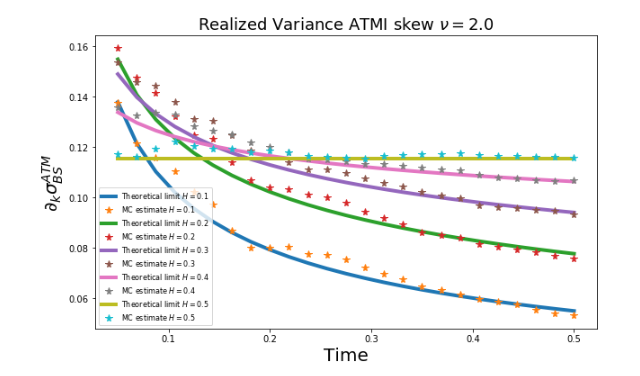}
\caption{ATMI skew short time limit for RV options in a rough Bergomi model with $\nu=2$.}
\label{fig:ATMIVarianceSkew}
\end{figure}
\begin{figure}[H]
\includegraphics[scale=0.7]{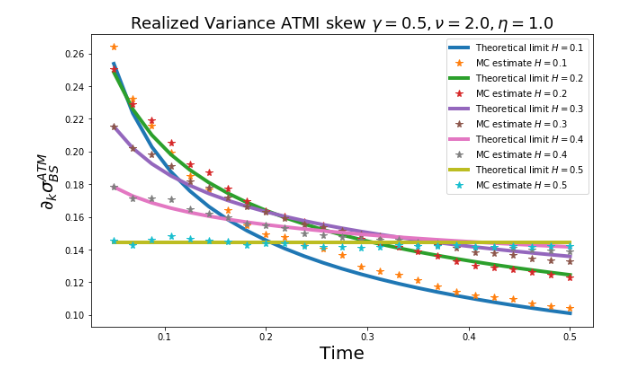}
\caption{ATMI skew short time limit for RV options in a mixed rough Bergomi model with $(\gamma,\nu,\eta)=(1/2,2,1)$.}
\label{fig:ATMIVarianceSkewMixed}
\end{figure}
\appendix
\section{Computations of the ATMI skew on VIX options}\label{sec:skewVIXOptions}
To ease the forthcoming computations we introduce the following functions
\begin{eqnarray*}G(H,\Delta,\beta)&&=\begin{cases} (2\beta)^{-2H}\Gamma_{low}(2H,2\beta\Delta) &\text{ if } \beta>0 \\
\displaystyle \frac{\Delta^{2H}}{2H} & \text{ if } \beta=0

\end{cases}\\
J(H,\Delta,\beta)&&=\begin{cases} \beta^{-H-1/2}\Gamma_{low}(H+1/2, \beta\Delta) &\text{ if } \beta>0 \\
\displaystyle \frac{\Delta^{H+1/2}}{H+1/2} & \text{ if } \beta=0

\end{cases}.
\end{eqnarray*}
Using Remark \ref{rem:Dsphi} we get
\begin{equation}\label{eq: MalliavinDerSigmaSquared}D_s\phi_u^2= \frac{2\phi_uD_s m(T,u)M_u^T}{(M_s^T)^2}-\frac{2\phi_um(T,u)D_sM_u^T}{(M_s^T)^2}.\end{equation}
We recall the reader that in the context of VIX options we have 
$$
\phi_t=\frac{1}{M_t^T2\Delta} E_t \left(\frac{1}{VIX_T}\int_T^{T+\Delta} (D_t v_s)ds\right). 
$$
and we have the limiting behaviour for $0\leq t\leq T$
\begin{eqnarray*}
\lim_{T\to 0}E (\phi_t)&&=\lim_{T\to 0}\frac{f'(Y_0)}{M_0^02\Delta VIX_0} \int_T^{T+\Delta} (r-t)^{H-1/2}\exp(-\beta(r-t)) dr\\&&=:\lim_{T\to 0}\frac{f'(Y_0)}{M_0^02\Delta VIX_0}K(T,\Delta,t)
\\&&=\frac{f'(Y_0)}{M_0^02\Delta VIX_0} J(H,\Delta,\beta).
\end{eqnarray*}
On one hand we have,
\begin{equation}\label{eq: auxB}E_s (D_sM_u^T)=m(T,s).\end{equation}
Hence, we may rewrite \eqref{eq: MalliavinDerSigmaSquared} under unconditional expectation as
\begin{equation}\label{eq: ExpectationMalliavinDerSigmaSquared}E\left(D_s\phi_u^2\right)=E \frac{2\phi_uD_s m(T,u)M_u^T}{(M_s^T)^2}-\frac{2\phi_u m(T,u)\phi_s}{M_s^T}.\end{equation}
On the other hand,
\begin{eqnarray*}
&&D_s m(T,u)\nonumber\\
&&=\displaystyle\frac{1}{2\Delta} D_sE_u\frac{1}{VIX_T}\int_T^{T+\Delta}(D_u(v_r)) dr\\
&&=\frac{1}{2\Delta}E_u\left(\frac{\int_T^{T+\Delta}D_s(D_u(v_r)) dr VIX_T-Ds(VIX_T)\int_T^{T+\Delta}(D_u(v_r)) dr}{(VIX_T)^2}\right)\\
&&=E_u\left(\frac{\int_T^{T+\Delta}D_s(D_u(v_r)) dr}{2\Delta VIX_T}-\frac{\int_T^{T+\Delta}(D_s(v_r)) dr\int_T^{T+\Delta}(D_u(v_r)) dr}{4\Delta^2\Delta(VIX_T)^3}\right)\\
&&=E_u\left(\frac{\int_T^{T+\Delta}D_s(D_u(v_r)) dr}{2\Delta VIX_T}-\frac{m(T,s)m(T,u)}{VIX_T}\right)
\end{eqnarray*}
Therefore we further develop \eqref{eq: ExpectationMalliavinDerSigmaSquared} into 
\begin{eqnarray}\label{eq:DsSigmaSquared}
E\left(D_s\phi_u^2\right)&=&E \left(\frac{-2\phi_u\int_T^{T+\Delta}D_s(D_u(v_r)) dr}{2\Delta VIX_T}-2\phi_u m(T,u)\phi_s\left(\frac{1}{M_s^T}+\frac{1}{VIX_T}\right)\right)\nonumber\\
&:=&A(T,s,u)+B(T,s,u).\end{eqnarray}
Taking limits for the expression $B(T,s,u)$ as in Theorem \ref{theadaptation} up to the normalizing terms and using the definition of $\phi_s$ along with \eqref{eq:DsSigmaSquared}  we obtain
\begin{eqnarray*}
&\displaystyle&\lim_{T\rightarrow 0} E  \displaystyle\int_0^T \phi_s \left(\int_s^T B(T,s,u) du\right)ds\\
=&\displaystyle&\lim_{T\rightarrow 0} E  \displaystyle\int_0^T \phi_s \left(\int_s^T \frac{-4\phi_u\phi_s m(T,u)}{M_0^0} du\right)ds\\
=&\displaystyle &\lim_{T\rightarrow 0}E  \displaystyle\int_0^T \frac{-4\phi_s^2}{M_s^T} \left(\int_s^T \phi_u m(T,u) du\right)ds\\
=&\displaystyle &\lim_{T\rightarrow 0}E \frac{-4(f'(Y_0))^2}{(M_0^0 2\Delta VIX_T)^2} \displaystyle\int_0^T \phi_s^2 \left(\int_s^T K(T,\Delta,u)^2 du\right)ds\\
=&\displaystyle &\lim_{T\rightarrow 0}E \frac{-4(f'(Y_0))^4}{(M_0^0 2\Delta VIX_T)^4} \displaystyle\int_0^T K(T,\Delta,s)^2 \left(\int_s^T K(T,\Delta,u)^2 du\right)ds\\
=&\text{ } & \frac{-4(f'(Y_0))^4}{(M_0^0 2\Delta VIX_0)^4} \displaystyle \frac{T^2}{2} (J(H,\Delta,\beta))^4
\end{eqnarray*}
where in the last step we use $T\geq u \geq s \geq 0$ along with continuity of $K(\cdot)$. Finally, using the ATMI result that gives
$$\lim_{T\rightarrow 0} u_0=\lim_{T\rightarrow 0} E\sqrt{ \frac{1}{T}\int_0^T \phi^2_s}ds=\frac{f'(Y_0)}{2\Delta VIX_0} J(H,\Delta,\beta),$$
we get
\begin{eqnarray}\label{limitB}\frac{1}{2}\displaystyle\lim_{T\rightarrow 0} E  \displaystyle\frac{\int_0^T \phi_s \left(\int_s^T B(T,s,u) du\right)ds}{u_0^3}&=\displaystyle \lim_{T\rightarrow 0}\displaystyle -\frac{1}{2\Delta} \displaystyle\frac{f'(Y_0)}{VIX_0 M_0^0}J(H,\Delta,\beta) T^2\quad
\end{eqnarray}
On the other hand, we proceed similarly with $A(T,s,u)$. First we have
\begin{eqnarray*}
&&\lim_{T\rightarrow 0} E\frac{1}{2\Delta}\frac{E_u\int_T^{T+\Delta}D_s(D_u(v_r)) dr}{VIX_T}\nonumber\\
&&=\frac{f''(Y_0)}{2\Delta VIX_0}\int_T^{T+\Delta}((r-u)(r-s))^{H-1/2}\exp(-\beta(2r-u-s))dr\\
&&=\frac{f''(Y_0)}{2\Delta VIX_0}I(T,\Delta,s,u).
\end{eqnarray*}
Moreover, due to the fact that $s<u<T$, we have
$$\lim_{T\rightarrow 0} E\frac{1}{2\Delta}\frac{E_u\int_T^{T+\Delta}D_s(D_u(v_r)) dr}{VIX_T}=\frac{f''(Y_0)}{2\Delta VIX_0}G(H,\Delta,\beta).$$
Carrying similar computations, as with $B(T,s,u)$, we obtain
\begin{eqnarray*}
&\displaystyle&\lim_{T\rightarrow 0} E  \displaystyle\int_0^T \phi_s \left(\int_s^T A(T,s,u) du\right)ds\\
&=\displaystyle&\lim_{T\rightarrow 0} E  \displaystyle\int_0^T \phi_s \left(\int_s^T \frac{2\phi_uD_s m(T,u)M_u^T du}{(M_s^T)^2}\right)ds\\
&=\displaystyle&\lim_{T\rightarrow 0} \frac{2}{M_0^0}E  \displaystyle\int_0^T \phi_s  \left(\int_s^T \phi_u D_sm(T,u) du\right)ds\\
&=\displaystyle&\lim_{T\rightarrow 0} \frac{2f''(Y_0)\Delta^{-1}}{M_0^02VIX_0}E  \displaystyle\int_0^T \phi_s  \left(\int_s^T \phi_u I(T,\Delta,u,s) du\right)ds\\
&=\displaystyle&\lim_{T\rightarrow 0} \frac{2f''(Y_0)f'(Y_0)\Delta^{-2}}{(M_0^0)^2 2VIX_0}E  \displaystyle\int_0^T \phi_s  \left(\int_s^T K(T,\Delta,u) I(T,\Delta,u,s) du\right)ds\\
&=\displaystyle&\lim_{T\rightarrow 0} \frac{2f''(Y_0)(f'(Y_0))^2\Delta^{-3}}{(M_0^0)^3 (2VIX_0)^3 } \displaystyle\int_0^T K(T,\Delta,s) \left(\int_s^T K(T,\Delta,u) I(T,\Delta,u,s) du\right)ds\\
&=\displaystyle&\lim_{T\rightarrow 0} \frac{f''(Y_0)(f'(Y_0))^2}{(M_0^0)^3 (2VIX_0)^3 } \Delta^{-3} (J(H,\Delta,\beta))^2 G(H,\Delta,\beta)T^2.
\end{eqnarray*}
were in the last step we use continuity of $I(\cdot)$ and $K(\cdot)$ and the fact $0\leq s \leq u \leq T$. Finally, recalling the ATMI result for $u_0^3$ we get
\begin{eqnarray}\label{limitA}\frac{1}{2}\displaystyle\lim_{T\rightarrow 0} E  \displaystyle\frac{\int_0^T \phi_s \left(\int_s^T A(T,s,u) du\right)ds}{u_0^3}&=\displaystyle\frac{G(H,\Delta,\beta)}{2J(H,\Delta,\beta)}\displaystyle\frac{f''(Y_0)}{f'(Y_0) } T^2
\end{eqnarray}

To result follows directly using Theorem \ref{theadaptation} along with expressions \eqref{limitA} and \eqref{limitB}.
\section{Computations of the ATMI skew sign in Heston}\label{VIXSkewHeston}
We notice that $E_T(v_r)=\theta+(v_T-\theta)\exp(-k(r-T))$. This implies that, for all $t<T$
\begin{eqnarray}
\label{MHestonlim}
&&M_t^T=E_t\sqrt{E_T\frac1\Delta\int_T^{T+\Delta}v_rdr}\nonumber\\
&&=E_t\sqrt{\theta+(v_T-\theta)\frac1\Delta\int_T^{T+\Delta}\exp(-k(r-T))dr}\nonumber\\
&&=E_t\sqrt{\theta+(v_T-\theta)\frac{1-\exp(-k\Delta)}{k\Delta}}\nonumber\\
&&\to E\sqrt{\theta+(v_0-\theta)\frac{1-\exp(-k\Delta)}{k\Delta}}\nonumber\\
&&=\sqrt{v_0}
\end{eqnarray}
in $L^2(\Omega)$, as $T\to 0$.
On the other hand, for all $s<t<T$
\begin{eqnarray}
\label{DMHeston}
D_sM_t^T&=&E_t\left[\frac{1}{2M_t^T}E_T\left(\frac1\Delta\int_T^{T+\Delta}D_sv_rdr\right)\right]\nonumber\\
&=&\frac{1}{2M_t^T}E_t\left(\frac1\Delta\int_T^{T+\Delta}D_sv_rdr\right).
\end{eqnarray}
Now, notice that
\begin{equation} \label{v_t}
v_{r} = v_{0}\exp{(-kr)} + \theta \int_{0}^{r} \exp{(-k(r-u)}) du +\nu \int_{0}^{r} \exp{(-k(r-u))} \sqrt{v_{u}} dWu, 
\end{equation}
which allows us to write
\begin{equation}
\label{derHeston}
D_sv_r=\nu\exp(-k(r-s))\sqrt{v_s}+\nu\int_{s}^{r} \exp{(-k(r-u))} D_s\sqrt{v_{u}} dWu.
\end{equation}
Consequently, it follows that
 $E_t(D_s v_r)\to \nu\exp(-kr)\sqrt{v_0}$ in $L^2(\Omega)$, as $T\to 0$.
Then we get
\begin{eqnarray}
\label{DMHestonlim}
&&D_sM_t^T\to\frac{\nu(1-\exp(-k\Delta)}{2k\Delta }
\end{eqnarray}
in $L^2(\Omega)$ as $T\to 0$. 
On the other hand, (\ref{mgeneral}) allows us to write
\begin{eqnarray}
\label{dermVIX}
D_sm(T,t)&=&\frac{1}{2\Delta} E_t \left(\frac{1}{VIX_T}E_T\int_T^{T+\Delta} (D_s D_t v_r)dr\right)\nonumber\\
&-&\frac{1}{4\Delta} E_t \left(\frac{1}{\Delta(VIX_T)^3}\int_T^{T+\Delta} ( D_t v_r)dr\int_T^{T+\Delta} ( D_s v_r)dr\right)\nonumber\\
&&=:T_1+T_2.
\end{eqnarray}
Let us first study the term $T_1$. It is easy to deduce from (\ref{derHeston}) that, for $s<t<r$
\begin{eqnarray}
&&D_sD_tv_r=\nu\exp(-k(r-t))D_s\sqrt{v_t}\nonumber\\
&&+\nu\int_{t}^{r} \exp{(-k(r-u))}D_t D_s\sqrt{v_{u}} dWu.
\end{eqnarray}
Lemma 5.3 in \cite{AE} gives us that, as $s,t\to 0$, $D_s\sqrt{v_t}\to \frac{\nu}{2}$ in $L^2(\Omega)$. This implies that
\begin{equation}
\label{T1}
E_T(D_sD_tv_r)\to\frac{\nu^2}{2}\exp(-kr),
\end{equation}
which in turn, gives
\begin{equation}
\label{T1}
T_1\to \frac{\nu^2(1-\exp(-k\Delta))}{4k\Delta \sqrt{v_0}}.
\end{equation}
Finally, we can write
\begin{equation}
\label{T2}
T_2\to -\frac{\nu^2(1-\exp(\Delta))^2}{4k^2\Delta^2 \sqrt{v_0}}.
\end{equation}
Then (\ref{DMHestonlim}), (\ref{DMHestonlim}), (\ref{T1}) and (\ref{T2}) give us that
\begin{eqnarray}
\label{skew_heston}
&&D_s m(T,t)M_t^T-m(T,t)D_sM_t^T\nonumber\\
&&\to \frac{\nu^2(1-\exp(-k\Delta))}{4k\Delta}-\frac{\nu^2(1-\exp(-k\Delta))^2}{4k^2\Delta^2 }\nonumber\\
&&-\left(\frac{\nu(1-\exp(-k\Delta)}{2k\Delta }\right)^2\nonumber\\
&&=\frac{\nu^2(1-\exp(-k\Delta))}{4k\Delta}-2\frac{\nu^2(1-\exp(-k\Delta))^2}{4k^2\Delta^2 }\nonumber\\
&&=\frac{\nu^2(1-\exp(-k\Delta))}{4k\Delta}\left(1-\frac{2(1-\exp(-k\Delta))}{k\Delta}   \right).
\end{eqnarray}

\section{Computations of the ATMI skew on realized variance options}\label{sec:skewVarianceOptions}
Using Remark \ref{rem:Dsphi} we get
$$D_s\phi_u^2= \frac{2\phi_uD_s m(T,u)M_u^T}{(M_s^T)^2}-\frac{2\phi_um(T,u)D_sM_u^T}{(M_s^T)^2}=:A(T,s,u)+B(T,s,u)$$
First we will develop $B(T,s,u)$. On one hand we have
\begin{equation}\label{eq: auxB}E_s (D_sM_u^T)=m(T,s).\end{equation}
Taking limits for the expression in Theorem \ref{theadaptation} up to the normalizing terms and using the definition of $\phi_s$ along with \eqref{eq: auxB}  we obtain
\begin{eqnarray*}
&\displaystyle&\lim_{T\rightarrow 0} E  \displaystyle\int_0^T \phi_s \left(\int_s^T B(T,s,u) du\right)ds\\
=&\displaystyle&\lim_{T\rightarrow 0} E  \displaystyle\int_0^T \phi_s \left(\int_s^T \frac{-2\phi_u\phi_s m(T,u)}{M_s^T} du\right)ds\\
=&\displaystyle &\lim_{T\rightarrow 0}E  \displaystyle\int_0^T \frac{-2\phi_s^2}{M_s^T} \left(\int_s^T \phi_u m(T,u) du\right)ds\\
=&\displaystyle &\lim_{T\rightarrow 0}\displaystyle \frac{-2(f'(Y_0))^2}{(M_0^0)^2(H+1/2)^2 T^2}\int_0^T\phi_s^2 \left(\int_s^T (T-u)^{2H+1}du \right) ds\\
=&\displaystyle &\lim_{T\rightarrow 0}\displaystyle \frac{-2(f'(Y_0))^2}{(M_0^0)^2(H+1/2)^2(2H+2) T^2}\int_0^T \phi_s^2(T-s)^{2H+2} ds\\
=&\displaystyle &\lim_{T\rightarrow 0}\displaystyle \frac{-2(f'(Y_0))^4}{(M_0^0)^4(H+1/2)^4(2H+2) (4H+4) T^4} T^{4H+4}\\
=&\displaystyle &\lim_{T\rightarrow 0}\displaystyle \frac{-2(f'(Y_0))^4}{v_0^4(H+1/2)^4(2H+2) (4H+4)} T^{4H}
\end{eqnarray*}
Finally, using the ATMI result that gives
$$\lim_{T\rightarrow 0} u_0=\lim_{T\rightarrow 0} E\sqrt{ \frac{1}{T}\int_0^T \phi^2_s}ds=\frac{|f'(Y_0)|}{(H+1/2)\sqrt{2H+2}v_0} T^{H-1/2},$$
we get
\begin{eqnarray*}\frac{1}{2}\displaystyle\lim_{T\rightarrow 0} E  \displaystyle\frac{\int_0^T \phi_s \left(\int_s^T B(T,s,u) du\right)ds}{u_0^3}&=\displaystyle \lim_{T\rightarrow 0}\displaystyle \frac{-f'(Y_0)(2H+2)^{3/2}}{v_0(2H+1)(2H+2) (4H+4)} T^{H+3/2}\\
&=\displaystyle\lim_{T\rightarrow 0}\displaystyle \frac{-f'(Y_0)}{v_0(2H+1)\sqrt{(2H+2)}} T^{H+3/2}.
\end{eqnarray*}
On the other hand, we proceed similarly with $A(T,s,u)$. First we have
\begin{eqnarray*}D_s m(T,u)&&=\displaystyle\frac{1}{T}\int_u^T E_u(D_s(D_u(v_r))) dr= \frac{1}{T}\displaystyle\int_u^TE_u\left((D_sf'(Y_r)(r-u)^{H-1/2} \right)dr\\
&&=\displaystyle\frac{1}{T}\int_{u}^Tf''(Y_r)((r-u)(r-s))^{H-1/2} dr=:I(T,u,s).
\end{eqnarray*}
\begin{remark}
Notice here that 
$$\lim_{T\rightarrow 0} E( I(T,u,s))=\frac{1}{T}f''(Y_0)\frac{(T-u)^{H+1/2}(u-s)^{H-1/2}}{H+1/2} F\left(\frac{T-u}{s-u}\right)$$
where $F(\cdot)=\text{}_2F_1\left(1/2-H,H+1/2,H+3/2,\cdot\right)$ denotes the Gaussian hypergeometric function.
\end{remark}
Carrying similar computations, as with $B(T,s,u)$, we obtain
\begin{eqnarray*}
&\displaystyle&\lim_{T\rightarrow 0} E  \displaystyle\int_0^T \phi_s \left(\int_s^T A(T,s,u) du\right)ds\\
&=\displaystyle&\lim_{T\rightarrow 0} E  \displaystyle\int_0^T \phi_s \left(\int_s^T \frac{-2\phi_uD_s m(T,u)M_u^T du}{(M_s^T)^2}\right)ds\\
&=\displaystyle&\lim_{T\rightarrow 0} \frac{-2}{M_0^0}E  \displaystyle\int_0^T \phi_s  \left(\int_s^T \phi_u I(T,u,s) du\right)ds\\
&=\displaystyle&\lim_{T\rightarrow 0} \frac{-2  f'(Y_0)}{(M_0^0)^2(H+1/2)T}E  \displaystyle\int_0^T \phi_s \int_s^T I(T,u,s)(T-u)^{H+1/2} du ds\\
&=\displaystyle&\lim_{T\rightarrow 0} \frac{-2 (f'(Y_0))^2 }{(M_0^0)^3(H+1/2)^2T^2}  \int_0^T (T-s)^{H+1/2} \int_s^T I(T,u,s)(T-u)^{H+1/2} du ds\\
&=\displaystyle&\lim_{T\rightarrow 0} \frac{-2 (f'(Y_0))^2f''(Y_0) }{(M_0^0)^3(H+1/2)^2T^3}\nonumber\\
&&\times  \int_0^T (T-s)^{H+1/2} \int_s^T \frac{(T-u)^{2H+1}(u-s)^{H-1/2}}{H+1/2} F\left(\frac{T-u}{s-u}\right) du ds\\
&=\displaystyle&\lim_{T\rightarrow 0} \frac{-2f'(Y_0))^2f''(Y_0) \mathcal{I}(H)}{v_0^3 (H+1/2)^2}  \;T^{4H}
\end{eqnarray*}
where we define $\mathcal{I}(H)$ as the finite limit
\begin{equation}\label{eq: auxA}\mathcal{I}(H)=\lim_{T\rightarrow 0}\displaystyle \frac{\displaystyle\int_0^T (T-s)^{H+1/2} \int_s^T \frac{(T-u)^{2H+1}(u-s)^{H-1/2}}{H+1/2} F\left(\frac{T-u}{s-u}\right) du ds}{T^{4H+3}}.\end{equation}
\begin{remark}
In practice, $\mathcal{I}(H)$ can be easily approximated using numerical integration packages and expression \eqref{eq: auxA} for small $T$.
\end{remark}
Finally, recalling the ATMI result for $u_0^3$ we get
\begin{eqnarray*}&&\frac{1}{2}\displaystyle\lim_{T\rightarrow 0} E  \displaystyle\frac{\int_0^T \phi_s \left(\int_s^T A(T,s,u) du\right)ds}{u_0^3}\\
&&=\displaystyle \lim_{T\rightarrow 0}\displaystyle \frac{f''(Y_0)(2H+2)^{3/2}(H+1/2)\mathcal{I}(H)}{f'(Y_0)} T^{H+3/2}.
\end{eqnarray*}
To conclude the proof we make use of Theorem \ref{theadaptation} and obtain
\begin{eqnarray*}
&\displaystyle\lim_{T\rightarrow 0}&\frac{\partial I_0^T}{\partial k}
(X_0)=\frac{1}{2}\lim_{T\rightarrow 0}E\frac{
\int_{0}^{T}\left( \phi_s \int_s^T D_{s}^W \phi_u^2 du \right) ds }{u _0^{3}T^{2}}\\
=&\displaystyle\lim_{T\rightarrow 0}&\frac{1}{2}\displaystyle\lim_{T\rightarrow 0} E  \displaystyle\frac{\int_0^T \phi_s \left(\int_s^T A(T,s,u) du\right)ds}{u_0^3} +E  \displaystyle\frac{\int_0^T \phi_s \left(\int_s^T B(T,s,u) du\right)ds}{u_0^3}\\
=&\displaystyle\lim_{T\rightarrow 0}&\left(\frac{f''(Y_0)\mathcal{I}(H)(2H+2)^{3/2}(H+1/2)}{f'(Y_0)}-\frac{f'(Y_0)}{v_0(2H+1)\sqrt{(2H+2)}}\right) T^{H-1/2}.
\end{eqnarray*}

\end{document}